\documentclass[numbook, envcountsect, envcountsame, envcountreset, runningheads, smallextended]{svjour3}\usepackage{graphicx}

\smartqed 

\usepackage{amsfonts}
\usepackage{amsmath}
\usepackage{amssymb}
\usepackage{url}
\usepackage{fancyhdr}
\usepackage{indentfirst}
\usepackage{enumerate}
\usepackage{dsfont}

\usepackage[misc]{ifsym}
\usepackage{graphicx}

\usepackage[colorlinks=true,citecolor=blue]{hyperref}

\usepackage{color}
\usepackage{grffile}

\usepackage[numbers]{natbib}
\usepackage{comment}

\renewcommand{\d}{\mathrm{d}}
\newcommand{\laweq}{\buildrel \mathrm{d} \over =}

\newcommand{\bQto}{\buildrel {Q^*} \over \rightarrow}
\newcommand{\Qito}{\buildrel {Q_i} \over \rightarrow}

\newcommand{\VaR}{\mathrm{VaR}}
\newcommand{\ES}{\mathrm{ES}}
\newcommand{\MES}{\mathrm{MES}}
\newcommand{\GES}{\mathrm{iMES}}
\newcommand{\VES}{\mathrm{rMES}}
\newcommand{\AES}{\mathrm{AES}}

\newcommand{\MVaR}{\mathrm{MVaR}}

\newcommand{\E}{\mathbb{E}}

\newcommand{\R}{\mathbb{R}}
\newcommand{\N}{\mathbb{N}}
\newcommand{\p}{\mathbb{P}}
\newcommand{\X}{\mathcal{X}}
\newcommand{\Y}{\mathcal{Y}}

\newcommand{\Q}{\mathcal{Q}}
\newcommand{\uQ}{{\underline{Q}}}
\newcommand{\bQ}{\mathcal{P}}

\newcommand{\one}{\mathbf{1}}
\newcommand{\id}{\mathds{1}}

\renewcommand{\ge}{\geqslant}
\renewcommand{\le}{\leqslant}

\renewcommand{\epsilon}{\varepsilon}

\newcommand{\trd}{}

\begin{document}

\title{Scenario-based Risk Evaluation}
\titlerunning{Scenario-based Risk Evaluation}

\author{Ruodu Wang \and Johanna F.~Ziegel}

\institute{R.~Wang \at Department of Statistics and Actuarial Science, University of Waterloo, Canada\\
 \email{wang@uwaterloo.ca} \and J.~.F.~Ziegel \at Institute of Mathematical Statistics and Actuarial Science, University of Bern, Switzerland\\ \email{johanna.ziegel@stat.unibe.ch}}

\date{Received: date / Accepted: date}

\maketitle

\begin{abstract}
Risk measures such as Expected Shortfall (ES) and Value-at-Risk (VaR) have been prominent in banking regulation and financial risk management. Motivated by practical considerations in the assessment and management of risks, including tractability, scenario relevance and robustness, we consider theoretical properties of scenario-based risk evaluation. We establish axiomatic characterizations of scenario-based risk measures that are comonotonic-additive or coherent, and we obtain a novel ES-based representation result.  We propose several novel scenario-based risk measures, including various versions of Max-ES and Max-VaR, and study their properties. 
The theory is illustrated with financial data examples.
\end{abstract}

\section{Introduction}\label{sec:1}

\subsection{Background}\label{sec:11}

Risk measures are used in various contexts in banking and insurance, such as  regulatory capital calculation, optimization, decision making, performance analysis, and risk pricing;  see e.g.~\citet{MFE15} for a general review of quantitative risk management. In practice, risk measures have to be estimated from data. Therefore, it is often argued that one has to use a \emph{law-based risk measure} (or a \emph{statistical functional}), such as a Value-at-Risk (VaR) or an Expected Shortfall (ES), both standard risk measures used in banking and insurance.

However, even assuming that the distribution of a risk is accurately obtained, it may not be able to comprehensively describe the nature of the risk. From the regulatory perspective, a regulator is more concerned about the behavior of a risk in an adverse environment, e.g.~during a catastrophic financial event; see e.g.~\citet{AER12} for related discussions. Only the distribution of the risk may not be enough to distinguish a potentially huge loss in a financial crisis from a potentially huge loss in a common economy but no loss in a financial crisis. As another simple example, the profit/loss from a lottery and that from an insurance contract may have the same distribution, but they represent very different types of risks and can have very different effects on the decision maker or the society. Therefore, it may be useful to evaluate a risk under \emph{different stress scenarios}. Summing up these evaluations in a single number would necessarily lead to a \emph{non-law-based risk measure}.

Finally, it is usually unrealistic to assume that the distribution of a risk may be accurately obtained. Model uncertainty is a central component of the current challenges in risk measurement and regulation, and its importance in practice  has been pivotal after the 2007 financial crisis (see e.g.~\cite{OCC11}) in both the banking  (e.g.~\cite{BASEL16}) and the insurance sectors (e.g.~\cite{IAIS14}).
Model uncertainty may be due to statistical/parameter uncertainty or more generally, structural uncertainty of the model or of the economic system.
A \emph{robust approach} should take into account the distribution of the underlying risk under several plausible model assumptions.

In the framework of Basel III \& IV \cite{BASEL16}, the standard risk measure for market risk is an Expected Shortfall ($\ES_p$) at level $p=0.975$. Thus, the Basel Committee on Banking Supervision has opted for a law-based risk measure. However, while ES is the basic building block for market risk assessment, the initial ES estimates are subsequently modified, in particular, two important adjustments are a \emph{stress adjustment} and a \emph{dependence adjustment} \cite[p.52 -- p.69]{BASEL16}, which then leads to the \emph{capital charge for modellable risk factors} (abbreviated as IMCC in \cite{BASEL16}).

The aim of this paper is to present a theoretical approach to the construction of risk measures that incorporates modifications such as a stress and dependence adjustment of an initial law-based risk measure into the risk measure itself. We call such risk measures \emph{scenario-based risk measures}; see Definition \ref{def:1}. Our approach has the advantage that the final result of the risk estimation can be understood theoretically and properties such as coherence and comonotonic additivity can be studied not only for the initial law-invariant risk measure but for the final risk measure that is the relevant output for further actions and decisions, such as the IMCC in the Basel III \& IV framework.

Before presenting our theoretical framework, let us give some details on the latest regulatory framework of the Basel Committee on Banking Supervision to illustrate how they deal with the issues mentioned above.

In the Fundamental Review of the Trading Book (FRTB) of Basel III \& IV \cite{BASEL16} for market risk, the time horizon is 10 days (two trading weeks), and each risk position (random loss) is modelled as a function of risk factors, such as equity prices, interest rates, credit spreads, and volatilities. Each risk factor is adjusted according to their category of liquidity (see \citet{LX18} for a precise mathematical formulation).
\trd{For simplicity, in the illustration below we consider a linear portfolio.}
Let $X=\sum_{i=1}^n  X_i$ be the aggregate portfolio loss at a given day, where $X_1,\dots,X_n$ are the corresponding risk factors in the aggregation (with weights included). Below we outline two adjustments that the  FRTB uses to calculate the regulatory capital.

\subsubsection*{(i) Stress adjustment}
\begin{enumerate}[(a)]
\item Specify a set $R$ of reduced risk factors which has a  sufficiently long history of observation (at least spanning back to and including 2007), such that
the ratio $$\theta=\max\left\{\frac{\ES_{F}(X)}{
\ES_{R}(X)},1\right\}$$
is less than 4/3, where $\ES_{F}(X)=\ES_p(\sum_{i=1}^n X_i)$ is the  current ES value calculated using all risk factors, and $\ES_{R}(X)=\ES_p(\sum_{i\in R}X_i)$ is the current ES value calculated using the reduced risk factors. The ratio $\theta$ is treated like a constant and only needs to be updated weekly.

\item Compute  $\ES$ for a model with the reduced risk factors, \emph{``calibrated to the most severe 12-month period of stress"},
and this is denoted by $\ES_{R,S}(X)$.
The period of ``most severe stress'', also called the \emph{stress scenario} corresponds to the rolling window of data of length one year that leads to the maximum possible value of ES using the reduced risk factor model \cite[p.6]{BASEL17}.
 Mathematically, $\ES_{R,S}(X)$ involves taking a maximum over a set $\Q$ of distributions estimated from sequences of data of length one year (many of them overlapping), namely
 $$
 \ES_{R,S}(X)=\max_{Q\in \mathcal Q}\ES_p^Q\left(\sum_{i\in R} X_i\right).
 $$

\item Use the formula
$$\widetilde{\ES}(X)=\ES_{R,S}(X) \times \theta$$
to get the stress-adjusted ES value.
\end{enumerate}
In particular, if the portfolio loss is modelled by only risk factors of sufficiently long history (spanning back to 2007), then $R=\{1,\dots,n\}$ and  the adjusted ES value is
 $$
\widetilde{\ES}(X) =\max_{Q\in \mathcal Q}\ES_p^Q\left(\sum_{i=1}^n X_i\right)= \max_{Q\in \mathcal Q}\ES_p^Q(X).
 $$

\subsubsection*{(ii) Dependence adjustment}
\begin{enumerate}[(a)]
\item Risk factors in the portfolio are grouped into a range of broad regulatory risk classes (interest rate risk, equity risk, foreign exchange risk, commodity risk and credit spread risk).
For the stress scenario (see (i)(b)), compute the ES of each risk class (according to (i)),
and denote their sum by
$\widetilde{\ES}_{C}(X)$.
By  {comonotonic-additivity} and {subadditivity} of ES (see Section \ref{sec:intro} for details),
this calculation is equivalent to using a model where all classes of risk factors are comonotonic (``non-diversified"), and it
represents the worst-case value of $\ES$ among all possible dependence structures (e.g.~\citet{EPRWB14}).
\item Use the formula
$$ \ES(X)= \lambda \widetilde{\ES}(X)+ (1-\lambda) \widetilde{\ES}_{C}(X), $$
where $\lambda$ is a constant (right now, $\lambda$ is chosen as $0.5$).  The quantity $\ES(X)$ is called the IMCC of the portfolio.
\end{enumerate}
Intuitively, the logic behind adjustment (i) is that risk assessment should be made based on stressed financial periods,
and that  behind  adjustment (ii) is that the dependence structure between risk factors is difficult to specify and a worst-case value is combined with the
original model to protect from overly optimistic diversification effects in the model specification. In addition to (i) and (ii), the IMCC value will finally be adjusted by using the maximum of its present calculation and a moving average calculation of 60 days times a constant  (currently 1.5). See \citet{EPRWB14,EWW15} for discussions on the aggregation of risk measures under dependence uncertainty.

In summary, in the FRTB, ES of the {same portfolio} is estimated under {different scenarios and models}:
stress (stressed, non-stressed), and
dependence (diversified, non-diversified).
These values are aggregated with mainly two operations (iteratively): {maximum} and {linear combination}.
In Theorem \ref{th:th1-prime}, we show that these two operations indeed are the two most crucial operations which lead to  coherence in the sense of \citet{ADEH99} for scenario-based risk measures.
 Section \ref{sec:52} contains a detailed data analysis for the stress adjustment (i) outlined above.

We briefly mention two other prominent examples of risk evaluation using scenarios.  First, the margin requirements calculation developed by the Chicago Mercantile Exchange \cite{CME10} relies on the maximum of the portfolio loss over several specified hypothetical scenarios \cite[p.63]{MFE15}. Our data example in Section \ref{sec:51} is similar to this approach.
The second example comes from the practice of credit rating, where a structured finance security (e.g.~a defaultable bond) is rated according to its behavior (conditional distributions) under each economic stress scenario. This approach, in different specific forms, appears in both the Standard and Poor's and Moody's rating methodologies \cite{SP09,M10}.

In this paper, we propose an axiomatic framework of scenario-based risk evaluation, which
has the three merits mentioned above, and is consistent with many existing risk measurement
procedures including the above examples. We shall keep the Basel formulas as our
primary example in mind.

\subsection{Our contribution and the structure of the paper}
In Section  \ref{sec:intro}, we introduce scenario-based risk measures. They include classic law-based risk measures, non-law-based risk measures such as the systemic risk measures CoVaR and CoES (\citet{AB16}), and many practically used risk calculation principles such as the Basel formulas for market risk, the margin requirements by the Chicago Mercantile Exchange, and the common rating measures used in credit rating, as mentioned above.
Axiomatic characterizations of scenario-based risk measures are studied in Section \ref{sec:3}.
In particular, we characterize scenario-based comonotonic-additive as well as coherent risk measures, where many surprising mathematical challenges emerge. 
In Section \ref{sec:messec}, we introduce  Max-ES and Max-VaR, and related families of risk measures and study their properties.  Data analyses are given in Section \ref{sec:6}, highlighting the broad range of possible interpretations of scenarios. In particular, scenario-based risk measures can be easily implemented for stress analysis and capital calculation. 

Our framework builds upon the axiomatic theory of coherent risk measures as pioneered by \citet{ADEH99}. A comprehensive review on risk measures can be found in the books of \citet{D12} and \citet{FS16}.

Distinguished from the literature, the main mathematical challenges in our paper come   from the novel framework of treating risk measures as functionals on a space of random variables (as in the traditional setting) through the distributions of a random variable under each scenario (probability measure). The mathematical structure is very different from the one with tuples of distributions as inputs, such as the classic framework of \citet{AA63} (see also \citet{GS89,CMMM13,HM16}) in decision theory under ambiguity.
The distributions of a random variable under each scenario are not arbitrary.  \citet{SSWW19} recently studied the issue of compatibility between distributions and scenarios; for instance, if $Q_1$ and $Q_2$ are mutually equivalent and the distribution of $X$ is uniform under $Q_1$, then it cannot be normal under $Q_2$. 
If tuples of distributions are used as inputs, the ``geometry" (interdependence) of the set of scenarios does not play a role in the characterization results. This is in sharp contrast to our framework, where the interdependence among the set of scenarios plays an important role in the characterization of risk measures. 
See the detailed discussion in Section \ref{sec:41}  and Example \ref{ex:mono}, as well as   Theorem \ref{th:th1-prime}, where the choice of the set of scenarios clearly matters for the characterization result.  

\citet{KPH13} and \citet{KP16}  studied properties of
risk measures based on scenarios from different angles than ours. Various scenario-based risk measures also appear in \citet{ZF09,ZKR12,AB16,R18} in different disguises. Our contribution is to study the consequence of the scenario-based property  instead of specific examples; thus, our results yield axiomatic support for specific risk measures in the above literature. 
For recent developments of risk measures, including various practical issues of statistical analysis, robustness, model uncertainty, and optimization, we refer to  \citet{FisslerZiegel2016,CF17,KSZ17,DE17,ELW17,WZ20}, and the references therein. 
   
\section{Scenario-based risk measures}\label{sec:intro}
 
\subsection{Definitions}

Let $(\Omega, \mathcal F)$ be a measurable space
and $\bQ$ be the set of all  probability measures on $(\Omega, \mathcal F)$. 
 For any probability measure $Q$ on $(\Omega, \mathcal F)$, write $F_{X,Q}$ for the cumulative distribution function (cdf) of a random variable $X$ under $Q$, i.e., $F_{X,Q}(x)=Q[X\le x]$ for $x\in \R$, and denote by $X\sim_Q F$  if $F=F_{X,Q}$.
 For two random variables $X$ and $Y$ and a probability measure $Q$, we write $X\laweq_Q Y$ if $F_{X,Q}=F_{Y,Q}$.
 For any cdf $F$, its generalized inverse is defined as
$F^{-1}(t)=\inf\{x\in \R: F(x)\ge t\}$ for $t\in (0,1]$.
Let $\X$ be the space of bounded random variables in $(\Omega,\mathcal F)$,
and $\Y$  be a convex cone of random variables containing $\X$, representing the set of random variables of interest, which are possibly unbounded.
We fix $\X$ throughout, whereas $\Y$ is specific to the functional considered. For instance, when considering the expectation $\E^Q$ for some $Q\in \bQ$, its domain $\mathcal Y$ is often chosen as the space of $Q$-integrable random variables, which depends on the choice of $Q$. However, it does not hurt to think of $\mathcal Y=\mathcal X$ for the main part of the paper.
 A  probability measure $\p\in \bQ$ shall be chosen as a reference probability measure in this paper, which may be interpreted as the real-world probability measure in some applications.

In this paper we use the term \emph{scenario} for a probability measure $Q\in \bQ$. The reason behind this choice of terminology is from the perspective of scenario analysis, as in the following example. This example will be referred to a few times throughout the paper.

\begin{example}\label{ex:ex21}
Let $\Theta$ be a random economic factor taking values in a set $K$ and $Q_\theta[\cdot] =\p[\cdot |\Theta=\theta]$, $\theta\in K$, be regular conditional probabilities with reference to $\Theta$.
The set $\{\Theta=\theta\}\in \mathcal F$ represents a possible economic event for each $\theta \in K$.
To analyze the behavior of a risk $X$ under each scenario  $\Theta=\theta$, $\theta \in K$,
 the respective distributions of $X$ under the probability measures $Q_\theta$ are of interest.
 \end{example}

Suppose that there is a collection $\Q$ of scenarios of interest.  As mentioned in the introduction, there may be different interpretations for the set $\Q$. In what follows, we take a collection of scenarios of interest and we shall not distinguish among the interpretations.
If a risk (random loss) $X$  and another risk $Y$ have the same distribution under all relevant scenarios in $\Q$,
 then they should be assigned identical riskiness, whatever sense of riskiness we speak of.
 This leads to the following definition of $\Q$-based mappings.

\begin{definition}\label{def:1}
For a non-empty collection of scenarios $\Q\subset \bQ$, a mapping $\rho:\Y\to (-\infty,\infty]$ is \emph{$\Q$-based} if $\rho(X)=\rho(Y)$ for $X,Y\in \Y$ whenever $X\laweq_Q Y$ for all $Q\in \Q$.\end{definition}

To put the above concept into risk management, we focus  on \emph{$\Q$-based risk measures}.
A risk measure is a mapping from $\Y$ to $(-\infty,\infty]$, with $\rho(X)<\infty$ for a  bounded $X$.
We use the term \emph{risk measure} in a broad sense, as it also includes deviation measures (such as variance) and other risk functionals. To keep things concise, our main examples are traditional risk measures such as VaR and ES, although our framework includes deviation measures. For the latter, see \citet{RUZ06}.
In this paper, we adopt the sign convention as in  \citet{MFE15}: for a risk $X\in \mathcal Y$, losses are represented by positive values of $X$ and profits are represented by negative values of $X$.

An immediate example of a $\Q$-based risk measure   is  one that depends on the joint law of a risk and an economic factor $\Theta$ as in Example \ref{ex:ex21}. 
{By writing $\Q=\{\p[\cdot|\Theta=\theta]:\theta\in K\}$, 
$\rho$ is $\Q$-based if and only if $\rho(X)$ is determined by the joint distribution of $(X,\Theta)$.} This setting includes the systemic risk measures CoVaR and CoES, which are evaluated based on conditional distributions of risks given events (see \citet{AB16}).
For a fixed random variable $S$ (the system) and $p\in (0,1)$,
  the systemic risk measure  {CoVaR} \trd{of the institution loss $X\in \mathcal Y$} is defined as:
$$\mathrm{CoVaR}_p^{S}(X)=\VaR^\p_p(S|X=\VaR^\p_p(X)),$$
and the other systemic risk measure  {CoES}  is defined as:
$$\mathrm{CoES}_p^{S}(X)=\E^\p[S|S\ge \mathrm{CoVaR}^S_p(X)],~~X\in \mathcal Y.$$
Since CoVaR and CoES are determined by the joint distribution of $(X,S)$, they are $\Q$-based risk measures for
$\Q=\{\p[\cdot|S=s]:s\in \R\}$.

Clearly, the $\Q$-based risk measures are generalizations of \emph{law-based} (single-scenario-based) risk measures, which are determined by the law of random variables in a given probability space.   Thus,
 $\Q$-based risk measures bridge law-based ones and generic ones, by noting the relationship (assuming $\p \in \mathcal Q$)
 $$\underbrace{\{\p\}}_{\mbox{law-based}}\subset \underbrace{\Q}_{\mbox{$\Q$-based}}\ \subset \underbrace{\mathcal P}_{\mbox{generic}}.$$
 Some immediate facts about $\Q$-based risk measures are summarized in the following.

\begin{enumerate}[(i)]
\item
All risk measures on $\Y$ are $\bQ$-based. In fact, if $X\laweq_Q Y$  for all $Q\in \bQ$, then $X=Y$.

\item If $\Q_1\subset\Q_2\subset \bQ$, then a $\Q_1$-based risk measure is also $\Q_2$-based.
\item
For  $\Q_1,\dots,\Q_n \subset \bQ$, let $\Q=\cup_{i=1}^n\Q_i$
and $\rho_i:\Y \to \R$ be $\Q_i$-based, $i=1,\dots,n$. For any $f:\R^n\to \R$, the mapping $f\circ(\rho_1,\dots,\rho_n):\Y\to \R$ is $\Q$-based.
\end{enumerate}

To see that claim (i) holds, let $\omega \in  \Omega$ and define $Q:\mathcal F\to \R$, $A\mapsto \id_A(\omega) $. One can verify that $Q$ defines a probability measure.
The distributions of $X$ and $Y$ under $Q$ are simply the point mass at $X(\omega)$ and $Y(\omega)$, respectively.
Therefore, $X\laweq_Q Y$ implies that $X(\omega)=Y(\omega)$.

Next we introduce a special type of collections of probability measures, which fits naturally into the context of Example \ref{ex:ex21}.
\begin{definition}
A collection of probability measures $\Q\subset \bQ$ is \emph{mutually singular}
if there exist mutually disjoint sets $A_Q\in \mathcal F$, $Q\in \Q$, such that
$Q[A_Q]=1$ for $Q\in \Q$. 
\end{definition}

An example of this type would be to take $Q_i[B]=\p[B|A_i] $ for $B\in \mathcal F$ where $A_1,\dots,A_n$ is a partition of $\Omega$ with $\p[A_i] > 0$ for $i=1,\dots,n$. That is, each $Q_i$ amplifies the probability of the events $A_i$ of interest, commonly seen e.g.~in importance sampling.
In Example \ref{ex:ex21}, $\Q=\{Q_\theta:\theta\in K\}$ is mutually singular.

We say that a tuple $(Q_1,\dots,Q_n) \in \bQ^n$ is mutually singular if $\{Q_1,\dots,Q_n\}$ is mutually singular, and any two of $Q_1,\dots,Q_n$ are non-identical.

\begin{remark}\label{rem:1}
In this paper, scenarios are treated  in a generic sense. They may have different interpretations in different contexts. In a statistical context, they may represent different values of an estimated parameter in the model of the risk. In  a simulation-based model, they may represent different parameters in the simulation dynamics, or  
{different probabilities used in importance sampling}. In a regulatory framework, they may represent different economic situations that the regulator is concerned about. 
In a financial market, to assess a contingent payoff, one may need to incorporate its distribution under the pricing measure and under the physical measure,  under multiple pricing measures, or with different heterogeneous opinions about the physical probability measure; these situations naturally require a risk measure determined by the distribution of the risk under different measures.
\end{remark}

\subsection{Preliminaries on risk measures}\label{sec:22}

We adopt  the terminology in \citet{ADEH99,K01,FS02}.
A risk measure $\rho:\Y\to (-\infty,\infty]$ is \emph{cash-invariant} if  it holds that $\rho(X+ c)=\rho(X)+c$ for $c\in \R$ and $X\in \Y$;
$\rho$ is \emph{monotone} if $\rho(X)\le \rho(Y)$ for $X,Y\in \Y$ with $X\le Y$;
$\rho$ is \emph{positively homogeneous} if $\rho(\lambda  X)=\lambda\rho(X)$ for $\lambda \in (0,\infty)$ and $X\in \Y$,
 and
$\rho$ is \emph{subadditive} if $\rho(X+Y)\le \rho(X)+\rho(Y)$ for $X,Y\in \Y$.
A risk measure is said to be \emph{monetary} if it is monotone and cash-invariant.
A risk measure is said to be \emph{coherent} if it is monetary, positively homogeneous and subadditive.
Two random variables $X$ and $Y$ in $(\Omega,\mathcal F)$ are \emph{comonotonic} if
$(X(\omega)-X(\omega'))(Y(\omega)-Y(\omega'))\ge 0$ for all $\omega,\omega'\in \Omega$.
A risk measure $\rho$ is \emph{comonotonic-additive} if $\rho(X+Y)=\rho(X)+\rho(Y)$ whenever $X$ and $Y$ are comonotonic.

Let us define some classic risk measures based on a single scenario $Q\in \bQ$.
The most popular risk measures in banking and insurance regulation are
the Value-at-Risk (VaR) and the Expected Shortfall (ES), calculated under a fixed  probability measure $Q\in \mathcal P$.
We shall refer to them as $Q$-VaR and $Q$-ES, respectively.
For these risk measures, their domain $\Y$ can be chosen as any convex cone of random variables containing $\X$, possibly the entire set of random variables.
For $p\in (0,1]$,  $\VaR^Q_p:\Y \to (-\infty,\infty]$ is defined as
 \begin{equation}\label{eq:vardef} \VaR^Q_p(X)= \inf\{x\in \R: Q(X\le x)\ge p\}=F_{X,Q}^{-1}(p),~~~~ X\in \Y,\end{equation}
 and  for $p\in (0,1)$, $\ES_p^Q:\Y\to (-\infty,\infty]$ is defined as 
 \begin{equation}\label{eq:esdef} \ES^Q_p(X)=\frac{1}{1-p}\int_p^1 \VaR_q^Q(X)\d q, ~~~~ X\in \Y.\end{equation}
Since $-\infty<\VaR_p^Q(X) \le \VaR_q^Q(X)<\infty$ for $p\le q<1$, the integral \eqref{eq:esdef} is well defined.
  In addition, we let $\ES_1^Q(X)=\VaR_1^Q(X)$.
  
For a specified scenario $Q$, $Q$-VaR and $Q$-ES belong to the class of distortion risk measures.
Define the following sets of functions $$\mathcal G=\{g\colon[0,1] \to [0,1] : g\mbox{ is increasing  with $g(0)=0$ and $g(1)=1$}\},$$
and   $\mathcal G_+=\{g\in \mathcal G: g \mbox{ is concave}\}.$ In this paper, the terms ``increasing'', ``decreasing'' and ``set inclusion'' are in the non-strict sense.
A \emph{$Q$-distortion risk measure} is defined as 
\begin{equation}\label{eq:eqdistdef}
\rho^Q_g(X)=\int_{-\infty}^0 (g\circ Q[X> x]-1 )\d x+\int^{\infty}_0 g\circ Q[X> x] d x,~~X\in \X_g.
\end{equation}
where $g\in \mathcal G$ is called the \emph{distortion function} of $\rho^Q_g$, and {$\X_g$ is the set of random variables such that the first integral in \eqref{eq:eqdistdef} is finite. Then, $\rho^Q_g \colon \X_g \to (-\infty,\infty]$ is a well-defined risk measure.} The set $\X_g$ always contains $\X$.
A \emph{$Q$-spectral risk measure} is a $Q$-distortion risk measure with a concave distortion function.
A $Q$-distortion risk measure is always monetary, positively homogeneous and comonotonic-additive.
A $Q$-spectral risk measure is, additionally, coherent.
$\VaR_p^Q$ has distortion function $g(x)=\id_{\{x>1-p\}},~x\in [0,1]$ and $\ES_p^Q$ has distortion function $g(x)=(1/(1-p))\min\{x,1-p\},~x\in [0,1]$.
For the above properties of distortion risk measures, see \citet[Section 4.7]{FS16}.

\section{Axiomatic characterizations}\label{sec:3} 
In this section, we establish axiomatic characterizations of $\Q$-based co\-mo\-no\-to\-nic-additive risk measures as well as $\Q$-based coherent risk measures.
We focus on a finite collection $\Q$   and  the set of bounded random variables, that is, $\Y=\X$. Focusing on the set of bounded random variables is not a limitation for 
 axiomatic characterization results, since properties on $\mathcal Y\supset \X$ imply those on $\X$.

Throughout this section,  $n$ is a  positive integer, and $\uQ=(Q_1,\dots,Q_n)$ is a vector of measures, where $Q_1,\dots,Q_n\in \mathcal P$ are (pre-assigned) probability measures on $(\Omega,\mathcal F)$, and
  $\Q=\{Q_1,\dots,Q_n\}$ is the set of these measures. The dimensionality of $\uQ$ and the cardinality of $\Q$ only differ if some of $Q_1,\dots,Q_n$ are identical.
    If $Q_1,\dots,Q_n$ are distinct, then the mutual singularity of  $\uQ$ is equivalent to that of $\Q$.
Write $\mathbf 0=(0,\dots,0)\in \R^n$ and $\mathbf 1=(1,\dots,1)\in \R^n$.
We say that $P\in\mathcal P$ dominates $\Q$, if $Q\ll P$ for all $Q\in \mathcal Q$, that is, if for all $Q \in \mathcal Q$, $Q$ is absolutely continuous with respect to $P$.
We say that $\uQ$ (or $\mathcal Q$) is \emph{atomless} if $(\Omega,\mathcal F,Q_i)$ is atomless for each $i=1,\dots,n$.  {Recall that  a probability space $(\Omega,\mathcal F,Q)$ is atomless
if there exists a uniform random variable $U$ on $(\Omega,\mathcal F,Q)$.}

\subsection{Novelty and challenges of our framework} \label{sec:41} 

We first illustrate  the distinction of our framework to other results in the literature,
as this distinction is quite subtle mathematically.
The main message is that the interdependence   {among $Q_1,\dots,Q_n$ (e.g.,~whether they are mutually singular or not)} matters  for the risk measure axioms in our framework, whereas this is  irrelevant for   results in the literature on scenario-based functionals (e.g.~\citet{CMMM13,KP16}).

The following simple example illustrates an interesting feature of scenario-based risk measures which is in sharp contrast to classic law-based risk measures. 
\begin{example}\label{ex:mono}
For $P, Q\in \mathcal P$, we define the $\{P,Q\}$-based risk measure $\rho$ as $\rho(X)=2\E^P[X]-\E^Q[X]$, $X\in \X$.
Note that $2P-Q \in \mathcal P$ if and only if  $2P \ge Q$, and under this condition, $\rho$ is the expectation under a probability measure $2P-Q$. If  $2P \ge Q$ fails to hold, then $\rho$ is not monotone. Hence,  $\rho$ is coherent  if and only if $2P \ge Q$.
\end{example}

To understand the implications of Example \ref{ex:mono}, we look at the notion of stochastic order. 
For $Q\in \mathcal P$ and two random variables $X, Y$, we write  
$X\prec_{\rm st}^Q Y$ if $F_{X,Q}(x)\ge F_{Y,Q}(x)$ for all $x\in \R$. We say that $\rho$ is \emph{$\Q$-monotone} if
$\rho(X)\le \rho(Y)$ for all $X,Y$ satisfying $X\prec_{\rm st}^Q Y$ for all $Q\in \Q$.
Since $X\le Y$ implies $X\prec_{\rm st}^Q Y$ for all $Q\in \Q$,
if a  risk measure $\rho$ is $\Q$-monotone,
then it is monotone.  
 As a well-known property, in the case of $\Q=\{P\}$ being a singleton, 
 a $\{P\}$-based risk measure 
  is monotone if and only if it is   $\{P\}$-monotone.
  However,  the risk measure $\rho$ in Example \ref{ex:mono} is generally not $\Q$-monotone (see Proposition \ref{prop:qmono}), but it is monotone and coherent if $2P \ge Q$.
This is in sharp contrast to the case of $\{P\}$-based risk measures.
 
The above observation suggests that the relationship among $P$ and $Q$ matters for the properties of $\rho$. To determine whether $\rho$ is a coherent risk measure, we need to specify two things: first, how $\rho$ incorporates the distributions of the risk under each scenario (i.e.~the mapping $(F_{X,Q})_{Q\in \Q}\mapsto \rho(X)$); second, how these scenarios interact with each other. 
In the case of $\{P\}$-based risk measures,   the mapping $F_{X,P}\mapsto \rho(X)$ solely determines properties of the risk measure,
whereas the choice of the measure $P$ is irrelevant. 
For instance, $\ES_p^P$ and $\E^P$ are always coherent risk measures regardless of the choice of $P$.

The above discussion is related to the popular notion of \emph{consequentialism} in decision theory in the framework of \citet{AA63}. 
In the framework of consequentialism,  two random outcomes $X$ and $Y$ (called Anscombe-Aumman acts) are compared via a preference model which aggregates the tuples of  distributions $(F_{X,Q})_{Q\in \Q}$ and $(F_{Y,Q})_{Q\in \Q}$; e.g.~the well-known robust preference of  \citet{GS89}. In the above framework, axioms are built on the set of tuples of distributions (e.g.~monotonicity is defined with respect to $\Q$-stochastic order) instead of the set of random variables. As a consequence, the set of measures $\Q$  does not play a role in the preference model. This is in sharp contrast to our framework.
For instance, Example \ref{ex:mono} is not allowed as a monotone preference in \citet{GS89}, whereas it is a coherent risk measure in the classic sense of \citet{ADEH99} assuming $2P\ge Q$. For risk management relevance, it is natural to impose economically relevant axioms on the set of random variables. 
Later we shall see that the above discussion plays a significant  role in the axiomatic characterization of scenario-based risk measures.

\subsection{Comonotonic-additive risk measures and Choquet integrals}\label{sec:32}

 As mentioned in Section \ref{sec:22}, the most popular class of risk measures in practice are the ones that are  additive for comonotonic risks.
We choose this class as the starting point to
establish an axiomatic theory of $\Q$-based risk measures.
It is well-known that \trd{law-based} monetary risk measures are closely related to the notion of Choquet integrals; for instance Yaari's dual utility functionals \cite{Y87} and Kusuoka representations \cite{K01} are based on Choquet integrals. 

\begin{definition}
A set function $c:\mathcal F\to \R$, is  \emph{increasing} if $c(A)\le c(B)$ for $A\subset B$, $A,B\in \mathcal F$, it is \emph{standard} if $c$ is increasing and satisfies $c(\varnothing)=0$ and $c(\Omega)=1$, and it is \emph{submodular} if
$$c(A\cup B)+c(A\cap B)\le c(A)+c(B), ~~A,B\in \mathcal F.$$
\end{definition}

\begin{definition}
 For a standard set function $c$ and  $X\in \mathcal X$, the \emph{Choquet integral} $\int X \d c$ is defined as
\begin{equation}
\label{eq:choquet1} \int X\d c =\int_{-\infty}^0 (c(X>x)-1)\d x + \int_0^\infty c(X>x) \d x.
\end{equation}
\end{definition}
The integral $\int X\d c$ in \eqref{eq:choquet1} might also be well-defined on sets larger than the set  $\X$ of bounded random variables.
Generally, depending on different choices of $c$, one may choose different domains for the Choquet integral.
A $Q$-distortion risk measure in \eqref{eq:eqdistdef} is exactly a Choquet integral by choosing $c=g\circ Q$.

Now we are ready to present the characterization for comonotonic-additive $\Q$-based risk measures, which is based on a celebrated result dating back to \citet{S86}.
{Because repeated appearances of some $Q_1,\dots,Q_n$ in $\mathcal Q$ matter for Theorem \ref{th:th2} but not for Definition \ref{def:1}, we will use both
 the vector $\uQ$ and the set $\mathcal Q$.}
\begin{theorem} \label{th:th2}
A risk measure $\rho$ on $\mathcal X$ is monetary (resp.~coherent), co\-mo\-no\-to\-nic-additive and $\Q$-based  if and only if
\begin{equation}\label{eq:qrep} \rho(X)=\int X\d \psi\circ \uQ,~~ X\in \mathcal X\end{equation}
for some   function $\psi:[0,1]^n\to [0,1]$ such that $\psi\circ \uQ$ is standard (resp.~$\psi\circ \uQ$ is standard and submodular).
\end{theorem}
 
\begin{proof}
Summarizing Theorems 4.88 and 4.94 in \citet{FS16}, 
a risk measure $\rho$ on $\mathcal X$ is monetary and comonotonic-additive if and only if $\rho$ is a Choquet integral for some standard set function $c$. In addition, $\rho$ is coherent if and only if $c$ is submodular. 

Suppose first that $\psi\circ\uQ$ is a standard set function. Then, by the result cited above, the right hand side of \eqref{eq:qrep} defines a comonotonic-additive and monetary risk measure $\rho$.  It is coherent if and only if $\psi\circ\uQ$ is additionally submodular.
From the definition of $\int X\d \psi\circ \uQ$, we have
\begin{equation}\label{eq:def-qdist}
\rho(X)= \int_{-\infty}^0 (\psi\circ \uQ [X>x]-1)\d x + \int_0^\infty \psi\circ \uQ[X>x] \d x,
\end{equation}
and hence $\rho$ is $\Q$-based.

Conversely, by the above representation result, $\rho$ can be written as a Choquet integral for some standard set function $c$. If $\rho$ is assumed to be coherent, then $c$ is additionally submodular. By taking $X=\id_A$, $A\in \mathcal F$, we have $c(A)=\rho(\id_A)$. Since $\rho$ is $\mathcal Q$-based, $\rho(\id_A)$ is determined by the distribution of $\id_A$ under $Q_1,\dots,Q_n.$
Let $R_\uQ \subset [0,1]^n$ be the range of $\underline Q$, that is, $R_\uQ=\{(Q_1[A],\dots,Q_n[A]):A\in \mathcal F\}$.
Since $\rho(\id_A)$ only depends on $\uQ[A]$, we can define $\psi:R_\uQ \to \R $ by 
$\psi(x_1,\dots,x_n) = \rho(\id_A)$
where $(x_1,\dots,x_n)=\uQ[A]$.
Thus, $c(A)=\rho(\id_A)=\psi\circ \uQ[A]$ for all $A\in \mathcal F$.
We can trivially extend the domain of $\psi$ to $[0,1]^n$ which does not affect the statement  that $c=\psi\circ \uQ$ is standard.
\end{proof}

We shall refer to a risk measure in \eqref{eq:qrep} as a \emph{$\uQ$-distortion risk measure}, which is, by Theorem \ref{th:th2}, precisely a monetary, co\-mo\-no\-to\-nic-additive and $\Q$-based risk measure.
Coherent $\uQ$-distortion risk measures are referred to as \emph{$\uQ$-spectral risk measures}.
For the  $\uQ$-distortion risk measure $\rho$ in \eqref{eq:qrep}, $\psi$ is called its \emph{$\uQ$-distortion function}, and it is unique on the range of $\uQ$, by noting that
$\rho(\id_A)=\psi\circ \uQ[A]$ for all $A\in \mathcal F$. The reliance on $\uQ$ is essential. For instance, taking $P,Q\in \mathcal P$ and defining $\rho(X)=(1/3)\E^{P}[X]+(2/3)\E^{Q}[X],~X\in \X$, then $\rho$ has a $(P,Q)$-distortion function and a $(Q,P)$-distortion function which are different.
The classes of $\uQ$-distortion and $\uQ$-spectral risk measures will be the building blocks of the theory of $\Q$-based risk measures.

Clearly, if  $n=1$, then the concepts of a $\uQ$-distortion risk measure, a $\uQ$-spectral risk measure and a $\uQ$-distortion function coincide with those defined for a single scenario in Section \ref{sec:22}. In that case, the representation  in \eqref{eq:qrep} reduces to
\begin{equation*}
\rho(X)=\int X\d \psi\circ Q_1,~ X\in \mathcal X\end{equation*}
where $\psi\in \mathcal G$ (and $\psi\in \mathcal G_+$ if $\rho$ is coherent).

The condition that $\psi\circ \uQ$ is standard or $\psi\circ \uQ$ is submodular may not be easy to verify in general, as it involves the joint
properties of $\psi$ and $\uQ$.
Next, we establish simple sufficient conditions based on solely $\psi$. These conditions are necessary and sufficient if $\uQ$ is mutually singular and atomless.

Recall that a function $f:[0,1]^n \to \R$ is called submodular if it holds for all $\mathbf{x},\mathbf{y} \in [0,1]^n$ that $f(\min(\mathbf{x},\mathbf{y})) + f(\max(\mathbf{x},\mathbf{y})) \le f(\mathbf{x}) + f(\mathbf{y})$, where $\min(\mathbf{x},\mathbf{y})$, $\max(\mathbf{x},\mathbf{y})$ denotes the componentwise minimum and maximum, respectively. By \citet[Theorem 3.12.2]{MS02}, 
the function $f$ is componentwise concave and submodular if and only if for all $\mathbf x,\mathbf y,\mathbf w,\mathbf z \in [0,1]^n$ with $\mathbf w \le \mathbf x,\mathbf y \le \mathbf z$ and $\mathbf w + \mathbf z = \mathbf x + \mathbf y$, we have
\begin{equation}\label{eq:2fak}
f(\mathbf x) + f(\mathbf y) \ge f(\mathbf w) + f(\mathbf z).
\end{equation}
In addition, if $f$ is two times continuously differentiable, then \eqref{eq:2fak} holds if and only if the entries of its Hessian are all non-positive. We call $f$ increasing if $\mathbf x \le \mathbf y$ implies $f(\mathbf x) \le f(\mathbf y)$.

\begin{proposition}\label{coro:31}
Let $\psi:[0,1]^n\to [0,1]$ be a function satisfying $\psi(\mathbf{0}) = 0$, $\psi(\mathbf{1}) = 1$.
\begin{enumerate}[(i)]
\item   If $\psi$ is increasing on the range of $\uQ$, then $\psi\circ \uQ$ is standard.
\item  If $\psi$ is increasing, componentwise concave, and submodular, then $\psi\circ\uQ$ is standard and submodular. More precisely, if $\psi$ is increasing and satisfies \eqref{eq:2fak} on the range of $\uQ$, then $\psi\circ\uQ$ is standard and submodular.
 \end{enumerate}
If $\uQ$ is mutually singular and atomless then the range of $\uQ$ is $[0,1]^n$, and the converse of (i) and (ii) are also true.
\end{proposition}
\begin{proof} 
Part (i) is trivial. 
For $A,B\in \mathcal F$ and $Q\in \mathcal Q$, we always have that $Q[A\cup B]+Q[A\cap B]= Q[A]+Q[B]$. Therefore, from \eqref{eq:2fak}, we obtain
$$\psi \circ \uQ [A\cup B]+\psi  \circ \uQ [A\cap B]\le \psi \circ \uQ[A]+\psi \circ \uQ[B]$$
which gives the submodularity of $ \psi \circ \uQ$, thus showing part (ii). 

If $\uQ$ is mutually singular and atomless, the map $\uQ:\mathcal{F}\to [0,1]^n$ is surjective.
Let  $A_1,\dots,A_n\in \mathcal F$ be disjoint sets such that $Q_i[A_i]=1$ for each $i=1,\dots,n$. For the converse of part (i), suppose that $\mathbf x, \mathbf y\in [0,1]^n$ with $x_1 \le y_1$ and $x_2 = y_2,\dots,x_n=y_n$. Let $B \in \mathcal{F}$ with $\mathbf x = (Q_1[B],\dots,Q_n[B])$. As $(A_1,\mathcal{F},Q_1)$ is an atomless probability space, there exists a set $C$ with $(B \cap A_1) \subset C \subset A_1$ and $Q_1[C] = y_1$ (\citet[Theorem 1]{D00}). We have $\mathbf y = (Q_1[C \cup B],\dots,Q_n[C \cup B])$, which yields the claim.
For the converse of part (ii), we show the claim for $n=1$, and the general case follows easily due to the fact that $\uQ$ is mutually singular. Let $x,y,w,z \in \mathbb{R}$ with $w \le x, y\le z$ and $ w+ z = x+y$. Take $B,C \in \mathcal{F}|_{A_1}$ with $Q_1[B] =  x$ and $Q_1[C] = y$. If $Q_1[B \cap C] >  w$, take $B' \subset (B\backslash C)$ with $Q_1[B'] = Q_1[B \cap C] -  w$ and $C' \subset (C\backslash B)$ with $Q_1[C'] = Q_1[B \cap C] - w$. Then, $\bar{C} = (C\backslash C') \cup B'$ fulfills $Q_1[\bar{C}] = y$ and $Q_1[B \cap \bar{C}] = w$. If $Q_1[B \cap C] < w$, take $B' \subset (B \cup C)^c$ with $Q_1[B'] = w - Q_1[B \cap C]$ and  $C' \subset C \cap B$ with $Q_1[C'] = w - Q_1[C \cap B]$. Then, $\bar{C} = (C\backslash C') \cup B'$ fulfills $Q_1[\bar{C}] = y$ and $Q_1[B \cap \bar{C}] = w$. The equation $w + z = x + y= Q_1[B] + Q_1[\bar{C}] = Q_1[B\cap \bar{C}]+ Q_1[B \cup \bar{C}]$, hence $z= Q_1[B \cup \bar{C}]$. Now, the submodularity of  $\psi\circ \uQ$ implies \eqref{eq:2fak}.
\end{proof}

Proposition \ref{coro:31} implies that it is straightforward to design various co\-mo\-no\-to\-nic-additive $\Q$-based risk measures by choosing increasing functions $\psi$.
We remark that, if $\uQ$ is not mutually singular, in order for $\psi\circ \uQ$ to be standard (resp.~submodular), it is generally not necessary for  $\psi$  to be increasing (resp.~componentwise concave and submodular). 
 In Example \ref{ex:mono}, the  distortion function of $\rho$ is $\psi:(s,t)\mapsto 2s-t$, which is not increasing; however $\rho$ is still a spectral risk measure if $2P \ge Q$. The following Proposition shows that in this example, $\rho$ cannot be $\Q$-monotone unless the range of $\uQ$ is degenerate, in the sense that it has empty interior which happens if $P = Q$.
\begin{proposition}\label{prop:qmono}
Let $\rho$ be a $\uQ$-distortion risk measure with $\uQ$-distortion function $\psi$. The risk measure $\rho$ is $\Q$-monotone if and only if $\psi$ is increasing on the range of $\uQ$.
\end{proposition}
\begin{proof}
If $\psi$ is increasing on the range of $\uQ$, the $\Q$-monotonicity of $\rho$ is immediate from \eqref{eq:def-qdist}. Conversely, suppose that $\mathbf x = \uQ[A] \le \mathbf{y} = \uQ[B]$ for some $A,B \in \mathcal{F}$. Then, $\one_A \prec_{\rm st}^Q \one_B$ for all $Q \in \Q$, hence, by $\Q$-monotonicity of $\rho$, we obtain $\psi(\mathbf{x}) =\psi \circ \uQ[A] = \rho(\one_A) \le \rho(\one_B) = \psi \circ \uQ[B] = \psi(\mathbf{y})$.
\end{proof} 
 
We proceed to discussing an integral representation of $\uQ$-distortion risk measures. 
In Section \ref{sec:22}, for a single scenario $Q$, a $Q$-distortion risk measure $\rho^Q_g$ is  defined as
\begin{equation}\label{eq:int-1}
\rho^Q_g(X)=\int_{-\infty}^0 (g\circ Q[X> x]-1 )\d x+\int^{\infty}_0 g\circ Q[X> x] d x,~~X\in \X.
\end{equation}
If $g$ is left-continuous,  $\rho^Q_g$ has a Lebesgue integral formulation  via an argument of integration by parts (\citet[Theorem 6]{DKLT12}), that is,
\begin{equation}\label{eq:int-2}
\rho_g^Q(X)=\int_0^1 \VaR^Q_{p}(X)\d \bar g(p),~~X\in \X,
\end{equation}
where $\bar g(t)=1-g(1-t)$ for $t\in [0,1]$. Note that in this case, $\bar g$ is right-continuous with $g(0)=1-g(1)=0$; thus $\bar g$ is a distribution function on $[0,1]$.
This property is key to the  integral representation in \eqref{eq:int-2}.
We establish an analogous integral formulation for the case of multiple scenarios under a similar assumption.
For a function $\psi:[0,1]^n\to [0,1]$, we define $\bar \psi (\mathbf u)=1-\psi(\mathbf 1-\mathbf u)$, $\mathbf u\in[0,1]^n$.

\begin{proposition}\label{th:minvargen}
\trd{Suppose that $\psi:[0,1]^n\to [0,1]$ is such that  $\bar \psi$ is  a distribution function on $[0,1]^n$.} Let  $\rho_\psi:\X\to \R$ be given by
\begin{equation}\label{eq:int-rep}
\rho_{\psi}(X)= \int_{[0,1]^n}\max\{\VaR_{u_1}^{Q_1}(X),\dots,\VaR_{u_n}^{Q_n}(X)\}\d \bar \psi(u_1,\dots,u_n).
\end{equation}
Then $\rho_\psi(X)$ is a $\uQ$-distortion risk measure with $\uQ$-distortion function $\psi$.
Moreover, if $\bar \psi$ is componentwise convex, then $\rho_{\psi}$ is a  $\uQ$-spectral risk measure.
\end{proposition}
\begin{proof}
Let 
$$
Y=\max\{F_{X,Q_1}^{-1}(U_1),\dots,F_{X,Q_n}^{-1}(U_n)\}=\max\{\VaR_{U_1}^{Q_1}(X),\dots,\VaR_{U_n}^{Q_n}(X)\},
$$
where $(U_1,\dots,U_n)\sim_{\p} \bar \psi$. For almost every $x\in \R$, we have
\begin{align*}
\p[Y\le x]&=\p[F_{X,Q_1}^{-1}(U_1)\le x,\dots,F_{X,Q_n}^{-1}(U_n)\le x]\\
&=\p[U_1\le F_{X,Q_1}(x),\dots, U_n\le F_{X,Q_n}(x)]\\&=\bar \psi ( Q_1[X\le x],\dots,Q_n[X\le x]) =1-\psi \circ \uQ [X>x].
\end{align*}
It follows that
\begin{align*}
 \rho_{\psi}(X)&=\E^\p[Y]=\int_{-\infty}^0 (\p[Y>x]-1)\d x+\int_0^\infty \p[Y>x]\d x\\
&=\int_{-\infty}^0 \left(\psi \circ  \uQ [X> x]-1\right)\d x+\int_0^\infty  \psi \circ \uQ [X> x] \d x = \int X\d \psi \circ \uQ .
\end{align*}
Note that  any distribution function $\bar \psi$ is increasing and supermodular. Hence, $\psi$ is increasing and submodular, and by Theorem \ref{th:th2} and Proposition \ref{coro:31}, we obtain the desired results.
\end{proof}

Proposition \ref{th:minvargen} provides a convenient way to construct various $\uQ$-distortion risk measures. For instance, one may choose $\bar \psi$ as an $n$-copula (\citet{J14}). A direct consequence of Proposition \ref{th:minvargen} is that any $\uQ$-distortion risk measure with $\uQ$-distortion function $\psi$ has a representation \eqref{eq:int-rep} if
$\bar \psi$ is a distribution function. 

For a single scenario $Q$, the distortion function $g$ of a $Q$-spectral risk measure $\rho_g^Q$ in \eqref{eq:int-1} is  concave,  implying that $\bar g$ is automatically a distribution function, and hence $\rho_g^Q$ always admits a representation in \eqref{eq:int-2}.
This property does not  carry through to the case of $\uQ$-distortion risk measures in general. More precisely,
the $\uQ$-distortion function of a $\uQ$-spectral risk measure is not necessarily always a distribution function,  because  all distribution functions on $[0,1]^n$ are supermodular but not vice versa. As a consequence, not all $\uQ$-spectral risk measure have representation \eqref{eq:int-rep}. This is in sharp contrast to the case of a single scenario.

\subsection{Coherent risk measures}
\label{sec:separable}

As a  classic result in the theory of risk measures, the Kusuoka representation \cite{K01}  states that, \trd{in an atomless probability space,} any single-scenario-based coherent risk measure admits a representation  as the supremum over a collection of spectral risk measures, which are mixtures of ES.

One naturally wonders whether a similar result holds true for $\Q$-based coherent risk measures.  
First, it is straightforward to notice  that a supremum over a collection of $\uQ$-spectral risk measure is always a $\Q$-based coherent  risk measure.
For the converse direction, we shall show that a $\Q$-based coherent risk measure admits a representation as the supremum of a collection of mixtures of $Q$-ES  for $Q\in \Q$, \trd{but this needs some non-trivial condition}.
More precisely,
a \emph{$\mathcal Q$-mixture of  ES}  is a risk measure $\hat \rho $
defined by
\begin{equation}\label{eq:maxes-rep2}
\hat \rho (X)=\sum_{i=1}^n w_i \int_0^1 \ES^{Q_i}_p(X)\d h_i(p),~X\in \X,
\end{equation}
 for some $\mathbf w=(w_1,\dots,w_n)\in   [0,1]^n$ with $\sum_{i=1}^n w_i=1$ and
 distribution functions  $  h_1,\dots,h_n $ on $[0,1]$.
  Clearly,  $\hat \rho $  is a $\uQ$-spectral risk measure, as each of the $Q$-ES  is a $\uQ$-spectral risk measure. Its $\uQ$-distortion function is given by $\psi(\mathbf{x}) = \sum_{i=1}^n w_i g_i(x_i),  ~ \mathbf{x} \in [0,1]^n,$
  where  for each $i=1,\dots,n$, $g_i$ is the distortion function of $\int_0^1 \ES^{Q_i}_p(\cdot) \d h_i(p)$, and thus, 
 $$ \psi(\mathbf{x}) = \sum_{i=1}^n w_i\left( 1 - h_i(1-x_i) + x_i\int_0^{1-x_i} \frac{1}{1-p}\d h_i(p)\right), \quad \mathbf{x} \in [0,1]^n.$$
  We denote by $\Phi_{\Q}$ the set of all $\mathcal Q$-mixtures of ES in \eqref{eq:maxes-rep2}. 
In the next theorem, we establish  that, if $\uQ$ is mutually singular \trd{and atomless}, then any $\Q$-based coherent risk measure $\rho$ can be written as 
a supremum of $\Q$-mixtures of ES, namely,
\begin{equation}\label{eq:maxes-rep}
\rho(X)= \sup_{\hat  \rho  \in  \Phi} \hat \rho (X),~~X\in \X,
\end{equation}
for some set $\Phi\subset \Phi_{\Q}$. 
Examples of risk measures of  type \eqref{eq:maxes-rep} will be discussed in Section \ref{sec:messec}.
\begin{theorem}\label{th:th1-prime}
\begin{enumerate}[(i)]
\item If $\rho:\X\to \R$  is the supremum of some  $\uQ$-spectral risk measures, then it is a  $\Q$-based coherent risk measure.
\item If $\uQ$ is mutually singular \trd{and atomless}, then a risk measure on $\X$ is a  $\Q$-based coherent risk measure if and only if it is a supremum of $\Q$-mixtures of ES   as  in \eqref{eq:maxes-rep}.
\end{enumerate}
 \end{theorem}
 
 Before proving Theorem \ref{th:th1-prime}, we establish some auxiliary results, which might be of independent interest. 
First, we discuss the Fatou property (\citet{D00, D12}), which we shall define with respect to a scenario dominating $\Q$.
Such a dominating scenario may be chosen as $Q^*=(1/n)\sum_{i=1}^n Q_i$.
Formally, a risk measure $\rho$ is said to satisfy the \emph{$\Q$-Fatou property}
if, for a uniformly bounded sequence $X_1,X_2,\dots\in \X$, the convergence $X_k \bQto X\in \X$  implies  $\rho(X)\le \liminf_{k\to\infty} \rho(X_k)$.
We also introduce a norm $\Vert\cdot\Vert_\Q$ on the $Q^*$-equivalence classes of $\X$,
defined as $\Vert\cdot\Vert_\Q=\sup\{x>0: Q^*[|X|>x]>0\}$, which is the usual $L^\infty$ norm for essentially bounded random variables in $(\Omega,\mathcal F,Q^*)$.
Note that in the definitions of the $\Q$-Fatou property and the norm $\Vert\cdot\Vert_\Q$, the dominating measure $Q^*$ can be chosen equivalently as any probability measure dominating $\Q$.
It is straightforward to check that all $\Q$-based monetary risk measures are continuous with respect to $\Vert\cdot\Vert_\Q$.
A quasi-convex risk measure $\rho$ is one that satisfies $\rho(\lambda X+(1-\lambda )Y)\le \max \{\rho(X),\rho(Y)\}$ for all $\lambda \in [0,1]$ and $X,Y\in \X$.

\begin{lemma}\label{lem:44}
If $\uQ$ is mutually singular, then a $\Q$-based quasi-convex risk measure that is continuous with respect to $\Vert\cdot\Vert_{\Q}$
satisfies the $\Q$-Fatou property.
\end{lemma}
\begin{proof}
Write  $Q^*=(1/n)\sum_{i=1}^n Q_i$, and note that  $X_k \bQto X\in \X$  implies
$X_k\Qito X$ for each $i=1,\dots,n$.
We shall show the lemma in a similar way to \citet[Theorem 30]{D12}, which states that a $\{Q^*\}$-based, $\Vert\cdot\Vert_{\{Q^*\}}$-continuous and quasi-convex functional satisfies the $\{Q^*\}$-Fatou property (first shown by \citet{JST06} with a minor extra condition).
A $\Q$-based risk measure is not necessarily $\{Q^*\}$-based, and hence the above result does not directly apply.
Nevertheless, we shall utilize \cite[Lemma 11]{D12}, which gives that for each $i=1,\dots,n$, $k\in \N$, there exist
a natural number $N_{k}$ and random variables
$Z^i_{k,1},Z^i_{k,2},\dots,Z^i_{k,N_{k}}$ having the same distribution as $X_k$ under $Q_i$,
such that $$\lim_{k\to \infty}\frac{1}{N_{k}}\sum_{j=1}^{N_{k}} Z^i_{k,j}=X \mbox{~~  in } \Vert\cdot\Vert_{\{Q_i\}}.$$
The numbers  $N_k$ can be chosen independently of $i$, as  explained by \citet[Remark 40]{D12}.
For $k\in \N$ and $j=1,\dots,N_k$,
let $Y_{k,j}= \sum_{i=1}^n Z_{k,j}^i \id_{A_i} $
 where $A_1,\dots,A_n\in \mathcal F$ are disjoint sets such that $Q_i[A_i]=1$ for $i=1,\dots,n$.
 It is clear that for each choice of $(i,j,k)$, $Y_{k,j}$ has the same distribution as $X_k$  under  $Q_i$, and
 $$\lim_{k\to \infty}\frac{1}{N_{k}}\sum_{j=1}^{N_{k}} Y_{k,j} =X \mbox{~~  in }  \Vert\cdot\Vert_\Q.$$
 Therefore, $\rho(Y_{k,j})=\rho(X_k)$. Finally, as  $\rho$ is $\Vert\cdot\Vert_\Q$-continuous, quasi-convex and $\Q$-based, we have
 $$
 \rho(X)=\lim_{k\to \infty}\rho\left(\frac{1}{N_{k}}\sum_{j=1}^{N_{k}} Y_{k,j}\right)\le \liminf_{k\to \infty} \max_{j=1,\dots,N_k}\rho(Y_{k,j})=\liminf_{k\to \infty} \rho(X_k).
 $$
 Thus, $\rho$ satisfies the $\Q$-Fatou property.
\end{proof}

As a direct consequence of  Lemma \ref{lem:44}, if $\Q$ is mutually singular, then any $\Q$-based coherent risk measure, such as a $\uQ$-spectral risk measure, satisfies the $\Q$-Fatou property. 
Next, we present a lemma which serves as a building block for the proof of Theorem \ref{th:th1-prime}. For $X\in \X$,  let $$L_X(\Q)=\{Y\in \X:Y \laweq_Q X ~\mbox{for all}~Q\in\Q\}.$$
That is, $L_X(\Q)$ is the set of all random variables identically distributed as $X$ under each measure in $\Q$.
Clearly $X\in L_X(\Q)$ and hence $L_X(\Q)$ is not empty.
\begin{lemma}\label{lem:42}
Suppose that $\uQ$ is mutually singular \trd{and atomless}, and the probability measure $P\ll(1/n) \sum_{i=1}^n Q_i$.  The functional $\rho:\X\to \R$ with $\rho(X)=\sup_{Y\in L_X(\Q)}\E^P[Y]$ is
a $\Q$-mixture of   ES.
\end{lemma}
\begin{proof}
Let $A_1,\dots,A_n\in \mathcal F$ be disjoint sets with $Q_i[A_i]=1$ for  $i=1,\dots,n$.
Write $Q^*=\frac1n\sum_{i=1}^n Q_i$ and $Z=\d P/\d {Q^*}.$
For each $i=1,\dots,n$, let $U_i$ be, under $Q_i$,  a uniform random variable on $[0,1]$
such that $Z=F^{-1}_{Z,Q_i}(U_i)$ $Q_i$-almost surely.
The existence of such a random variable $U_i$ can be guaranteed, for instance, by \citet[Lemma A.32]{FS16}.
By the Fr\'echet-Hoeffding inequality (\citet[Remark 3.25]{R13}),
for $Y\in \X$, we have  $\E^{Q_i}[ZY]\le \E^{Q_i}[ZF_{Y,Q_i}^{-1}(U_i)]$.
It follows that, for $Y\in L_X(\Q)$,
\begin{align*}
\E^P[Y] =\frac{1}{n}\sum_{i=1}^n\E^{Q_i}\left[\frac{\d P}{\d {Q^*}}Y \right]
 \le \frac{1}{n}\sum_{i=1}^n \E^{Q_i}\left[ZF_{X,Q_i}^{-1}(U_i) \right].
\end{align*}
 On the other hand, it is easy to verify that  $  \sum_{i=1}^n  F_{X,Q_i}^{-1}(U_i) \id_{A_i}\in L_X(\Q)$,
 and$$ \E^P\left[ \sum_{i=1}^n  F_{X,Q_i}^{-1}(U_i) \id_{A_i}\right]=\frac{1}{n}\sum_{i=1}^n \E^{Q_i}\left[ZF_{X,Q_i}^{-1}(U_i) \right].$$
Therefore,
$$\sup_{Y\in L_X(\Q)}\E^P[Y]= \frac{1}{n}\sum_{i=1}^n \E^{Q_i}\left[ZF_{X,Q_i}^{-1}(U_i) \right].$$
Note that
$$\E^{Q_i}\left[ZF_{X,Q_i}^{-1}(U_i) \right]=\int_0^1 F^{-1}_{Z,Q_i}(u)F^{-1}_{X,Q_i}(u)\d u,$$
and the function $\bar g:[0,1]\to [0,1]$, $t\mapsto \int_0^t F^{-1}_{Z,Q_i}(u)\d u$ is in $\mathcal G$ and is convex.
It follows that the mapping $X\mapsto \E^{Q_i}[ZF_{X,Q_i}^{-1}(U_i) ]$ is  a spectral risk measure in the form of \eqref{eq:int-2}.
Therefore, $\rho$ is a linear combination of  $Q$-spectral risk measures, $Q\in \Q$. Note that each $Q$-spectral risk measure is a mixture of $Q$-ES (\citet[Theorem 4]{K01}), and hence $\rho$ is a $\Q$-mixture of  ES.
\end{proof}

\begin{proof}[Proof of Theorem \ref{th:th1-prime}]
Concerning part (i), it suffices to observe that $\uQ$-spectral risk measures are coherent, and that a supremum of $\Q$-based coherent risk measures is also a $\Q$-based coherent risk measure.

For part (ii), since $\rho$ is coherent, by Lemma \ref{lem:44}, it has the $\Q$-Fatou property.
From the classic coherent   risk measure representation (\citet{D00}),
there exists a set $  \mathcal R\subset \mathcal P$ of probability measures which are absolutely continuous with respect to $Q^*$, such that
\begin{equation}\label{eq:th1}
\rho(X)=\sup_{P\in \mathcal R}\E^P[X],~~X\in \X. \end{equation}
Now fix $X\in \X$. As $\rho$ is $\Q$-based, $\rho(Y)=\rho(X)$ for all $Y\in L_X(\Q)$.
It follows that
\begin{align*}
\rho(X)&=\sup_{Y\in L_X(\Q)} \sup_{P\in \mathcal R}\E^P\left[Y\right]= \sup_{P\in \mathcal R}\sup_{Y\in L_X(\Q)}   \E^P [ Y].
\end{align*}
By Lemma \ref{lem:42}, for each $P\in \mathcal R$, the map $\X\to \R$, $X\mapsto \sup_{Y\in L_X(\Q)}   \E^P [ Y]$  is a mixture of $Q$-ES for $Q\in \Q$.
Therefore, $\rho$ is the supremum of $\Q$-mixtures of ES.
 \end{proof}
 
The representation in \eqref{eq:maxes-rep} resembles the risk measure  in the Basel FRTB formula; see Section \ref{sec:1}. Indeed, it is remarkable that only using maximums and linear combinations  of $Q$-ES, as done in \cite{BASEL16}, one arrives at all possible $\Q$-based coherent risk measures, if $\Q$ is mutually singular.
Certainly, the set of measures $\Q$ chosen in \cite{BASEL16} is not necessarily mutually singular, allowing for more possible forms of coherent risk measures.  Nevertheless,     the maximum of linear combinations  of $Q$-ES is the only form that is coherent for all choices of scenarios, among the general class of scenario-based risk measures.

The characterization   in Theorem \ref{th:th1-prime} on the set of bounded random variables can be
extended to general $L^q$-spaces.   
In the following, let $Q_0\in \mathcal P$ be a probability measure dominating $\Q$,
and $L^q(Q_0)$, $q\ge 1$ be the set of random variables with a finite $q$-th moment under $Q_0$.
\begin{proposition}\label{th:5}
Suppose that $\uQ$ is mutually singular and \trd{atomless}, and $q\ge 1$. The mapping $\rho:L^q(Q_0)\to \R $ is a  $\Q$-based coherent risk measure if and only if it is a supremum of $\Q$-mixtures of ES as  in \eqref{eq:maxes-rep}.  
\end{proposition}
\begin{proof} 
Clearly, the supremum of $\Q$-mixtures of ES is a $\Q$-based coherent risk measure.   
It suffices to show the ``only-if" statement. 
First, is it easy to verify, with the same arguments as in Theorem \ref{th:th1-prime}, that  Lemma  \ref{lem:42}  holds true with ``$\X\to \R$" replaced by ``$L^q(Q_0)\to \R\cup\{\infty\}$".
Note that $\rho$ in Proposition \ref{th:5} is finite-valued, and all finite-valued convex risk measures on $L^q$ are continuous (\citet[Theorem 7.24]{R13}). 
Therefore, $\rho$ admits a representation in \eqref{eq:th1} with $\X$ replaced by $L^q(Q_0)$. 
Using the same argument in the proof of Theorem \ref{th:th1-prime} (ii), we conclude that 
$\rho$ is the supremum of $\Q$-mixtures of ES.
\end{proof}
\begin{remark}
In Proposition \ref{th:5}, we have formulated $\rho$ such that it is real-valued on $L^q(Q_0)$. 
This assumption is useful, as not all $\Q$-mixtures of ES are finite on $L^q(Q_0)$, and this is because integrability is not preserved among equivalent measures. \end{remark}

 Theorem \ref{th:th1-prime} is  the most technical result of this paper. The mutual singularity of $\uQ$ is used repetitively in the proof  and it  is not dispensable for part (ii).   
 Without the mutual singularity of $\uQ$, the converse of Theorem \ref{th:th1-prime} (ii) is not true. 
 Note that the representation \eqref{eq:maxes-rep} is $\Q$-monotone, and we have seen in Example \ref{ex:mono} that there are $\Q$-based coherent risk measures which are not $\Q$-monotone, see also Proposition \ref{prop:qmono}. 
 
We expect that an interesting characterization result for $\Q$-based coherent risk measures without the assumption of mutual singularity of $\uQ$ will at least require the additional assumption of $\rho$ being $\Q$-monotone, possibly even for a weaker stochastic order than the usual stochastic ordering. 
However, even with such an additional assumption, the characterization problem without assuming mutual singularity is wide open, and it seems beyond the reach of current technical tools. The current line of proof requires the understanding of the functional defined in Lemma \ref{lem:42}. The functional is easy to understand via the Hardy-Littlewood inequality in the case of mutual singularity. However, without this assumption, the functional is hard to understand and may not be the right tool for a characterization result.


\begin{remark}\label{rem:distinguish}
  \citet{KP16} considered  $\Q$-based risk measures   $\rho$ with the representation 
$    \rho(X)=f(\rho_1(X),\dots,\rho_n(X))  $
for some aggregation function $f:\R^n\to \R$ and risk measures  $\rho_1,\dots,\rho_n$ each  based on one scenario; see Eq.~(22) of that paper.  They imposed some   axioms on   $f$ and $\rho_1,\dots,\rho_n$.
Our axioms are directly imposed on the risk measure $\rho$ and we do not assume a particular functional form.
\end{remark}

\section{Scenario-based VaR and ES and other examples}\label{sec:messec}

Because of the prominent importance of VaR and ES in external regulatory capital calculation and internal risk management, we investigate several examples of scenario-based risk measures which can be seen as natural generalizations VaR and ES in a multi-scenario framework. In this section, $\Y$ is any convex cone of random variables containing $\X$.
 

 

   \begin{example}[Max-ES and Max-VaR] Let $p \in (0,1)$. For a collection of measures $\Q$,  the \emph{Max-ES (MES)} is defined as
\begin{equation}\label{eq:MESdef} \MES_p^{\Q}(X)=\sup_{Q\in \Q}\ES_p^{Q}(X),~~X\in \Y,\end{equation}
and the \emph{Max-VaR (MVaR)}  is defined as
\begin{equation*} \MVaR_p^{\Q}(X)=\sup_{Q\in \Q}\VaR_p^{Q}(X),~~X\in \Y.\end{equation*}
   \end{example}
Max-ES and Max-VaR incorporate information evaluated under each scenario, and make a \emph{conservative} capital calculation by taking the maximum. We call them \emph{max-type} risk measures. They appear in the literature of robust optimization; see e.g.~\citet{ZF09,ZKR12}. 
Similar to the single-scenario-based ES and VaR in \eqref{eq:vardef} and \eqref{eq:esdef}, Max-ES and Max-VaR have different mathematical properties. 

Max-ES is known to be coherent (\citet{ZF09}), but in general, and in contrast the single-scenario-based ES, it fails to be comonotonic-additive; see Example \ref{ex:mes-nca} in Appendix \ref{sec:appendix}. Quite surprisingly, the risk measure $\MVaR_p^\Q$ satisfies comonotonic-additivity. If $\Q$ is finite with $n$ elements, this follows by choosing $\bar{\psi}$ as the distribution function of the Dirac measure at $(p,\ldots,p) \in [0,1]^n$ in Proposition \ref{th:minvargen}. For general sets $\Q$ of scenarios, see Appendix \ref{app:pr5}.

\begin{remark} As a classic result (\citet{D00}), a coherent risk measure $\rho$ on $\X$ with the Fatou property has a dual representation
\begin{equation*}
 \rho(X)=\sup_{Q\in \Q}\E^Q[X],~~ X\in \X,\end{equation*}
 for some set of probability measures $\Q$.
Clearly, $\rho$ is  a  max-type $\Q$-based risk measure with a different building block (the expectation) than Max-VaR and Max-ES.
This includes, in particular, the methodology for margin requirement calculation developed by Chicago Mercantile Exchange \cite[p.63]{MFE15}. In addition, \citet{DSV15}  feature the maximum-type of risk measures in the context of calculation of initial margins for bilateral portfolios.
\end{remark}


 Other than the max-type,  using a finite  or  continuous mixture  is also a convenient and simple way to construct $\Q$-based risk measures  {(however, using a infinite mixture requires a measure specified on $\mathcal Q$, which is not always easy)}. 
Below we present a few other ways to formulate ES in the framework of $\Q$-based risk measures. One may define the corresponding versions of VaR or any other law-based risk measure, but we take $\ES$ as a main example in this section due to its relevance in Basel III \& IV.

 \begin{example}
Take $p\in(0,1)$ and a  finite $\Q=\{Q_1,\dots,Q_n\}$, and
define the  \emph{Average-ES} (AES) as  the average of ES across different scenarios, that is,
\begin{equation}\label{eq:mix-es-a}
\AES_{p}^{\Q}(X)= \frac{1}n \sum_{i=1}^n  \ES_p^{Q_i}(X),~~X\in \Y.
\end{equation}
  It is obvious that $\AES_{p}^{\Q}$ is a coherent and comonotonic-additive risk measure. Its $\uQ$-distortion function $\psi(\mathbf{u}) =  (n(1-p))^{-1}\sum_{i=1}^n \min\{u_i,1-p\}$ is increasing, component\-wise concave and submodular. For $n=2$, the associated function $\bar{\psi}(\mathbf u) = 1 - \psi (1-\mathbf u)$, $\mathbf u \in [0,1]^n$ is a componentwise convex distribution function on $[0,1]^2$, so in this case, $\AES_{p}^{\Q}$ has an integral representation by Proposition \ref{th:minvargen}. For $n \ge 3$, $\bar \psi$ fails to be a distribution function.
  \end{example}

\begin{example}
Recall that  the single-scenario-based ES in \eqref{eq:esdef} is an average of VaR of probability level beyond $p\in (0,1)$.
Utilizing this connection, we define
 the \emph{integral  Max-ES} (iMES) as the integral of MVaR, that is,
\begin{equation}
\label{eq:mes-var-def}
{\GES}_p^\Q(X)=\frac{1}{1-p}\int_p^1 \MVaR_q^\Q(X)\d q,~~X\in \Y.
\end{equation}
For finite $\Q$ with $n$ elements, we can choose $\bar \psi$ as the distribution function of a uniform distribution over the diagonal line segment $$\{(u_1,\dots,u_n)\in [p,1]^n: u_1=u_2=\dots=u_n\}$$ to obtain the integral Max-ES as a special case of Proposition \ref{th:minvargen}. Its $\uQ$-distortion function $\psi$ is given by $\psi(\mathbf u)=\min\{\max\{u_1,\dots,u_n\},1-p\}/(1-p)$, $\mathbf u \in [0,1]^n$. This verifies that ${\GES}_p^\Q$ is comonotonic-additive.
However, $\psi$ is not componentwise concave, which  implies that ${\GES}_p^\Q$ is not a coherent risk measure for   mutually singular $\uQ$ by Proposition \ref{coro:31}. For random variables $X$ such that $F_{X,Q_i}$, $i=1,\dots,n$, are stochastically ordered, it holds that $\GES_p^\Q(X) = \MES_p^\Q(X)$.

Noting that $ \MVaR_q^\Q(X)$ is increasing in $q$, 
one may  equivalently write
\begin{equation}
\label{eq:eqMES2} {\GES}^\Q_p(X)=\ES_p^{\p}\left(\MVaR_U^\Q(X)\right) = \ES_p^{\p}\left(\sup_{Q\in\Q} F^{-1}_{X,Q}(U)\right),~~X\in \Y,\end{equation}
where $U\sim_\p \mathrm{U}[0,1]$. 
\end{example}

\begin{example}
Another way to utilize the ES and a maximum operator is via independent replications of $X$ under different scenarios.
 Let $p\in (0,1)$, $Q_1,\dots,Q_n$ be distinct scenarios and $\Q=\{Q_1,\dots,Q_n\} $.
Define  a \emph{replicated  Max-ES} (rMES) as the ES of  a maximum of independent copies, that is,
\begin{equation}
\label{eq:eqMES3}
{\VES}^\Q_p(X)=\ES_p^{\p}\left(\max_{i=1,\dots,n} X_i\right),~~X\in \Y,\end{equation}
where $X_i\sim_\p F_{X,Q_i}$, $i=1,\dots,n$, and $X_1,\dots,X_n$ are independent under~$\p$.
The risk measure  ${\VES}^\Q_p$ is defined for  a finite collection $\Q$ so that the maximum in \eqref{eq:eqMES3} is well-posed. The risk measure  ${\VES}^\Q_p$ finds some similarity to MINVAR in \citet{CM09}; see Example \ref{ex:ex31} for more details. The replicated Max-ES grows to the maximum of the essential supremum of the distributions $F_{X,Q_i}$ when the number $n$ of scenarios goes to infinity. Therefore, it is likely to be too conservative or even plainly uninformative if $n$ is too large; see also Section \ref{sec:6}.

The replicated Max-ES is comonotonic-additive and coherent. We can determine its $\uQ$-distortion function as follows. Suppose that $X_1,\dots,X_n$ are independent under $\p$, $X_i\sim_\p F_{X,Q_i}$, $i=1,\dots,n$, and $U\sim_\p \mathrm{U}[p,1]$.
The distribution function $F_{\max}$ of $\max \{X_1,\dots,X_n\}$ under $\p$ is 
 $x\mapsto \prod_{i=1}^n Q_i[X\le x]$.
Moreover, the survival function of $F_{\max}^{-1}(U)$ under $\p$ is given by
\begin{align*}
 x\mapsto 1-\frac{(F_{\max}(x)-p)_+}{1-p} &=1- \frac{( \prod_{i=1}^n Q_i[X\le x]-p)_+}{1-p}\\&= \frac{\min\{1- \prod_{i=1}^n Q_i[X\le x],1\} }{1-p}.
 \end{align*}  
 Hence, by letting $\psi: \mathbf u \mapsto \min\{1-\prod_{i=1}^n (1-u_i),1-p\}/(1-p)$,
 we have $ \psi \circ \underline Q[X>x]= \p(F^{-1}_{\max} (U)>x)$, $x\in \R$ and hence
 $$
 \int X \d \psi \circ \underline Q = \E^\p\left(F^{-1}_{\max} (U)\right)  = \ES^\p_p\left(\max_{i=1,\dots,n} X_i\right) =\VES^{\mathcal Q}_p(X).
 $$
 Thus, $\psi$ is the $\underline Q$-distortion function of $\mathrm{iMES}^{\mathcal Q}_p$ which is componentwise concave and submodular. 
 
We find that $\bar\psi(\mathbf u) = (\prod_{i=1}^n u_i - p)_+/(1-p)$, $\mathbf u \in [0,1]^n$ is a distribution function, so the replicated Max-ES an integral representation as in Proposition \ref{th:minvargen}. Indeed, $\bar\psi$ arises when combining uniformly distributed marginals on $[p,1]$ with an Archimedean copula with generator $$[0,1] \to [0,\infty),\quad t \mapsto -\log((1-p)t + p).$$ This generator is completely monotone, so we obtain a valid copula for any dimension $n$ (\citet{McNeilNeslehova2009}).
\end{example}

 {The risk measures ${\GES}^\Q_p,$ and ${\VES}^\Q_p$  are connected through the fact that if $X_1,\dots,X_n$ in \eqref{eq:eqMES3} are comonotonic instead of independent under $\p$, then \eqref{eq:eqMES3}  gives rise to \eqref{eq:eqMES2}.}
 Each of the $\Q$-based risk measures
$\MES^\Q_p,$ $\AES_p^\Q$, $ {\GES}^\Q_p$ and ${\VES}^\Q_p$  may be seen as a natural generalization of the single-scenario-based risk measure $\ES_p^Q$.
Although bearing similar ideas,  these risk measures have different properties and values.
If $\Q=\{Q\}$, then the above five risk measures are all equal. They are generally non-equivalent and satisfy an order summarized below. Proposition \ref{th:th32} summarizes the results of this section; some are already explained above. 

 \begin{proposition}\label{th:th32}
Let  $\Q$ be a finite collection of scenarios and $p\in (0,1)$.
\begin{enumerate}[(i)]
\item  $\MES_p^\Q$ is coherent, but generally  not  comonotonic-additive.
  \item $\MVaR_p^\Q$ is  comonotonic-additive, positively homogeneous and monetary, but generally  not   coherent.
\item $\AES^\Q_p$ is comonotonic-additive and coherent.
 \item ${\GES}^\Q_p$ is comonotonic-additive, but generally not coherent.
  \item ${\VES}^\Q_p$ is comonotonic-additive and coherent.
\item  $\AES^\Q_p(X)\le \MES^\Q_p(X)\le{\GES}^\Q_p (X)\le {\VES}^\Q_p(X)$ for all $X\in \mathcal Y$.
\item If $\Q$ is a singleton,  then all inequalities in (vi) are equalities. 
  \end{enumerate}
\end{proposition} 

\begin{proof}
It only remains to show (vi). The first inequality is trivial. For the second inequality, observe that for each $Q \in \Q$ and $U\sim_\p \mathrm{U}[0,1]$, we have $\ES^\Q_p(X) = \ES^\p_p(F_{X,Q}^{-1}(U))$, hence the claim follows by using \eqref{eq:eqMES2}. The third inequality follows by observing that the $\uQ$-distortion function of $\GES$ is pointwise smaller than the $\uQ$-distortion function of $\VES$. This holds since $\max\{u_1,\dots,u_n\} \le 1 - \prod_{i=1}^n(1-u_i)$ is equivalent to $\prod_{i=1}^n(1-u_i) \le (1-u_j)$ for all $j=1,\dots,n$ where the latter clearly holds for any $\mathbf u \in [0,1]^n$.
\end{proof}

 The above examples illustrate that the framework of $\Q$-based risk measures is flexible, and it allows for a great variety of risk measures to be formulated, even simply from the ES and a fixed $p$. Parts (i),(ii) and (iv) of Proposition \ref{th:th32} hold true if $\Q$ is infinite; see Appendix \ref{app:pr5}. 

There is a simple relationship between iMES (MVaR) and ES (VaR) when the collection of scenarios $\Q$ is the economic scenarios in Example \ref{ex:ex21}.
\begin{proposition}\label{prop:mes-es}
Let $\Q=\{Q_\theta:\theta\in K\}$ be as in Example \ref{ex:ex21}. For $p\in (0,1)$,
 $\MVaR_p^\Q(X) \ge \VaR_p^\p (X)$ and $\GES_p^\Q(X) \ge \ES_p^\p (X)$ for all $X\in \mathcal Y$.
\end{proposition}
\begin{proof}
We only need to show that $\MVaR_p^\Q(X) \ge \VaR_p^\p (X)$, which implies $\GES_p^\Q(X) \ge \ES_p^\p (X)$.
Take $x<  \VaR_p^\p (X)$, which implies $\p[X\le x]<p$.
As $\p$ is a convex combination of $Q_\theta$, $\theta\in K$, for some $\theta\in K$, we have $Q_\theta [X\le x]<p$,
implying $x \le \VaR^{Q_\theta}_p(X)\le  \MVaR^\Q_p(X)$. Therefore,
$$\MVaR^\Q_p(X)\ge \sup\{x\in \R: x<\VaR_p^\p (X)\}=\VaR_p^\p(X).$$
\end{proof}

Proposition \ref{prop:mes-es} implies that when using the economic scenarios in Example \ref{ex:ex21}, iMES (MVaR) is more conservative than ES (VaR) over the unconditional real world probability measure $\p$.
Note that $\MES^\Q_p(X)\ge \ES^\p_p(X)$ does not hold in general (see Example \ref{ex:mes-es-ct} in Appendix \ref{sec:appendix} for a counter-example), although this inequality almost always holds empirically, as we shall see in the data analysis in Section \ref{sec:6}. 

Below, we discuss two more connections of our risk measures to the ones in the literature. 
\begin{example}
Recently, \citet{R18} studied 
combinations of risk measures.
For a finite collection of scenarios $\{Q_1,\dots,Q_n\}$, these risk measures take the form $f(\rho^{Q_1},\dots,\rho^{Q_n})$ for some function $f$, where $\rho^{Q_i}$ is $\{Q_i\}$-based, $i=1,\dots,n$, and this includes the max-type risk measures.
Another  scenario-based risk measure of the type  $f(\rho^{Q_1},\dots,\rho^{Q_n})$ is given by \citet{KP16}, defined as
$$\rho(X)= \sup_{(w_1,\dots,w_n)\in \mathcal W}\left\{\sum_{i=1}^n w_i \rho^{Q_i}_{h_i}(X)\right\},~~X\in \X,$$
where $\rho^{Q_i}_{h_i}$, $i=1,\dots,n$ are $Q_i$-distortion risk measures  in \eqref{eq:eqdistdef} and $\mathcal W$ is a subset of the standard simplex $\{(w_1,\dots,w_n)\in [0,1]^n: \sum_{i=1}^n w_i=1\}$.
\end{example}

\begin{example}[Scenario-based MINVAR]\label{ex:ex31} In \eqref{eq:int-rep}, by choosing $\bar \psi(\mathbf u)=\prod_{i=1}^n u_i$ for $\mathbf u\in [0,1]^n$, we obtain
\begin{equation*} \rho(X)=\E^\p[\max\{X_1,\dots,X_n\}],~~X\in\X,\end{equation*}
where  $F_{X_i,\p}= F_{X,Q_i}$ for $i=1,\dots,n$, and $X_1,\dots,X_n$ are independent under $\p$.
Then, $\rho$ is a  $\uQ$-spectral  risk measure  with corresponding $\uQ$-distortion function $\psi(\mathbf u)=1-\prod_{i=1}^n (1-u_i)$, $\mathbf u\in [0,1]^n.$
 The  risk measure $\rho$ is coherent.  The single-scenario-based risk measure MINVAR (\citet{CM09}), defined as
 $$\mathrm{MINVAR}(X)=\E^\p[\max\{X_1,\dots,X_n\}],~~X\in\X,$$
 where $X_1,\dots,X_n$ are iid copies of $X$ under $\p$, is a special case of $\rho$ by choosing $Q_1=\dots=Q_n=\p$.
\end{example}

\section{Data analysis for scenario-based risk measures}\label{sec:6} 
{In this section, we discuss two examples of data analysis for $\Q$-based risk measures. The two examples are conceptually different with the aim to illustrate different possible interpretations for the collection $\Q$ of scenarios; cf.~Remark \ref{rem:1}. Various versions of the $\Q$-based Expected Shortfalls as in Section \ref{sec:messec} are chosen to illustrate the main ideas; clearly the analysis may be applied to other scenario-based risk measures.}

\subsection{Max-Expected Shortfalls for economic scenarios}\label{sec:51}

Taking up Example \ref{ex:ex21}, we consider $Q_i=\p[\cdot|\Theta=\theta_i]$, $i=1,\dots,4$,  where $\Theta$ is an economic factor taking values in $\{\theta_1,\dots,\theta_4\}$,  where $\p$ can be interpreted as the real-world probability measure. While the $\Q$-based Expected Shortfalls of $X$ are clearly defined mathematical quantities, it is not completely obvious how to estimate them. The approach we describe can be justified under suitable assumptions on the data generating processes. However, we leave a detailed study of the proposed estimator for future work.

Here, we assume that we have four data sets $D_1 = \{X_{1}^{Q_1},\dots,X_{N_1}^{Q_1}\}$, \dots, $D_4 = \{X_{1}^{Q_4},\dots,X_{N_4}^{Q_4}\}$ such that the empirical distribution of $D_i$ is a reasonable estimate of $F_{X,Q_i}$, $i=1,\dots,4$. Then, we estimate the risk measures $\ES_p$, $\MES_p$, $\GES_p$, and $\VES_p$ given at \eqref{eq:MESdef}, \eqref{eq:mix-es-a}, \eqref{eq:mes-var-def}, and \eqref{eq:eqMES3}, respectively, by their empirical counterparts.

Given a series of returns $(X_t)_{t \in \mathbb{N}}$, for each trading day, we compute $\Q$-based expected shortfalls of $X_t$ estimated on a rolling window of length $250$. The four scenarios can be interpreted as 
\begin{multline*}\{\theta_1,\dots,\theta_4\}\\=\{\text{high volatility},\text{low volatility}\}\times \{\text{good economy},\text{bad economy}\}.\end{multline*} The value of $\Theta$ is based on the values of VIX (high volatility/low volatility) and S\&P 500 (good economy/bad economy). To be precise, for day $t_0$ we use the time window $t_0-250,\dots,t_0-1$. Then, we use the VIX to split the time period into two categories depending on whether the VIX is higher or lower than its empirical median in the time window. We removed a log-linear trend from the S\&P 500 since 1950, and then we subdivide the $125$ days with high volatility in the current time window into two categories of (almost) equal size according to whether the S\&P 500 residuals are above or below their median during those $125$ days. The same is then done for the $125$ days with low volatility. This results in a split of the time window into four scenarios of (almost) equal size.

The sets $D_1$, \dots, $D_4$ consist of the values of $X_t$ for $t = t_0-250,\dots,t_0-1$ depending on which scenario the respective day has been assigned.
We considered return data from the NASDAQ Composite Index, the DAX Index, Apple Inc.~stock, and BMW stock. The data are freely available and were obtained from \url{https://finance.yahoo.com}. The considered time periods are 1991--2018. We do not consider data from before 1990 because there is no VIX data available. We chose the confidence level $p = 0.9$ for simplicity. For each series of return data, we also computed the empirical $\ES_p$ using a rolling windows the same size. The results of the analysis are summarized in Figure \ref{fig:2}.

\begin{figure}[h]
\begin{center}
\includegraphics[width=0.48\textwidth]{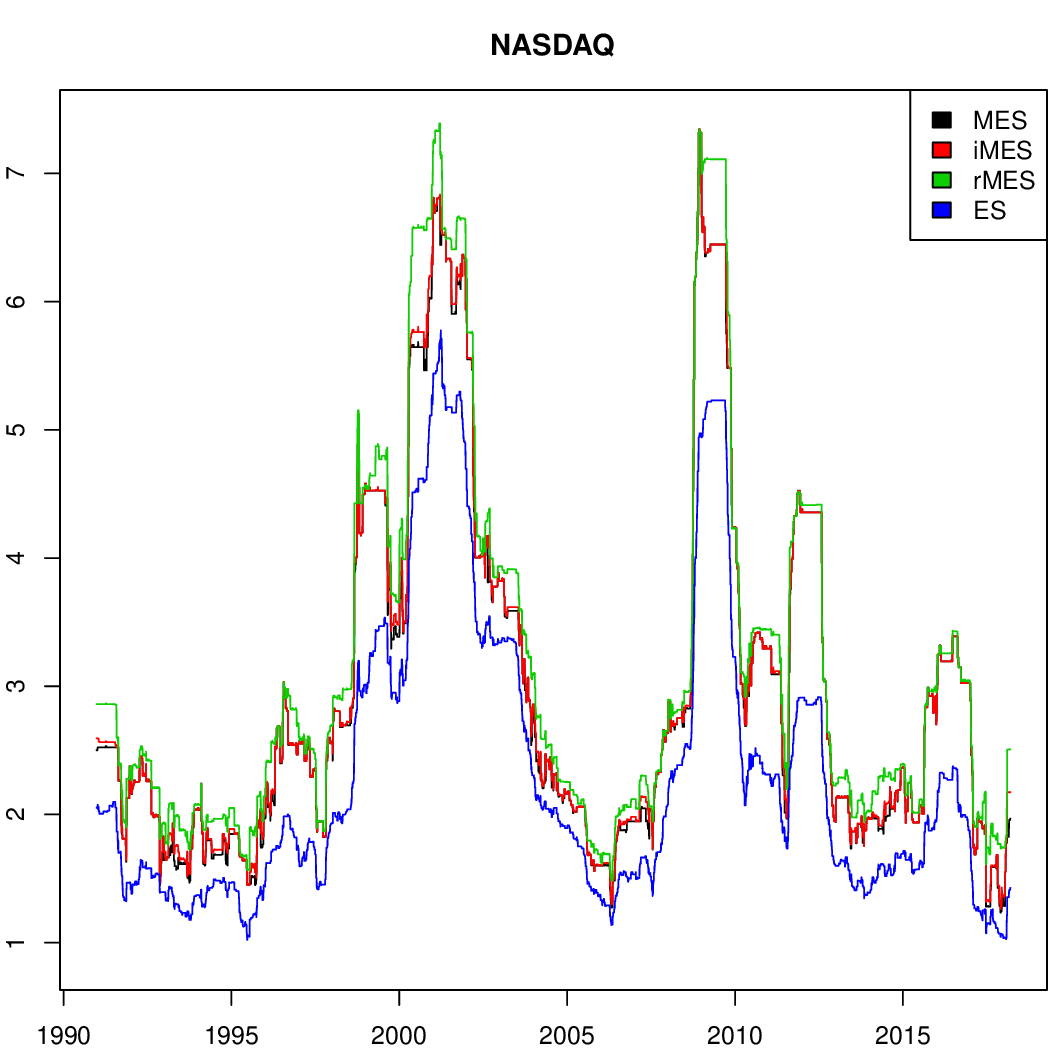}
\includegraphics[width=0.48\textwidth]{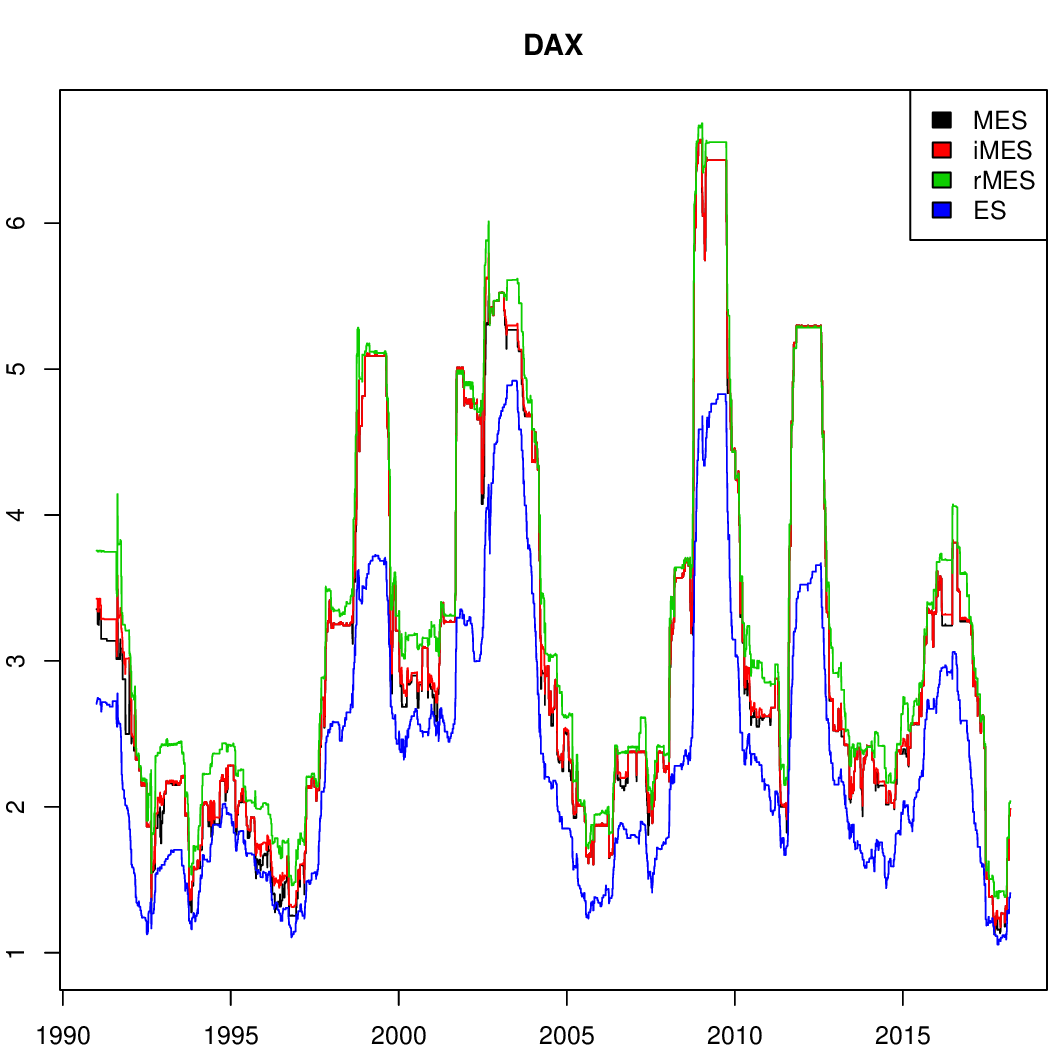}
\includegraphics[width=0.48\textwidth]{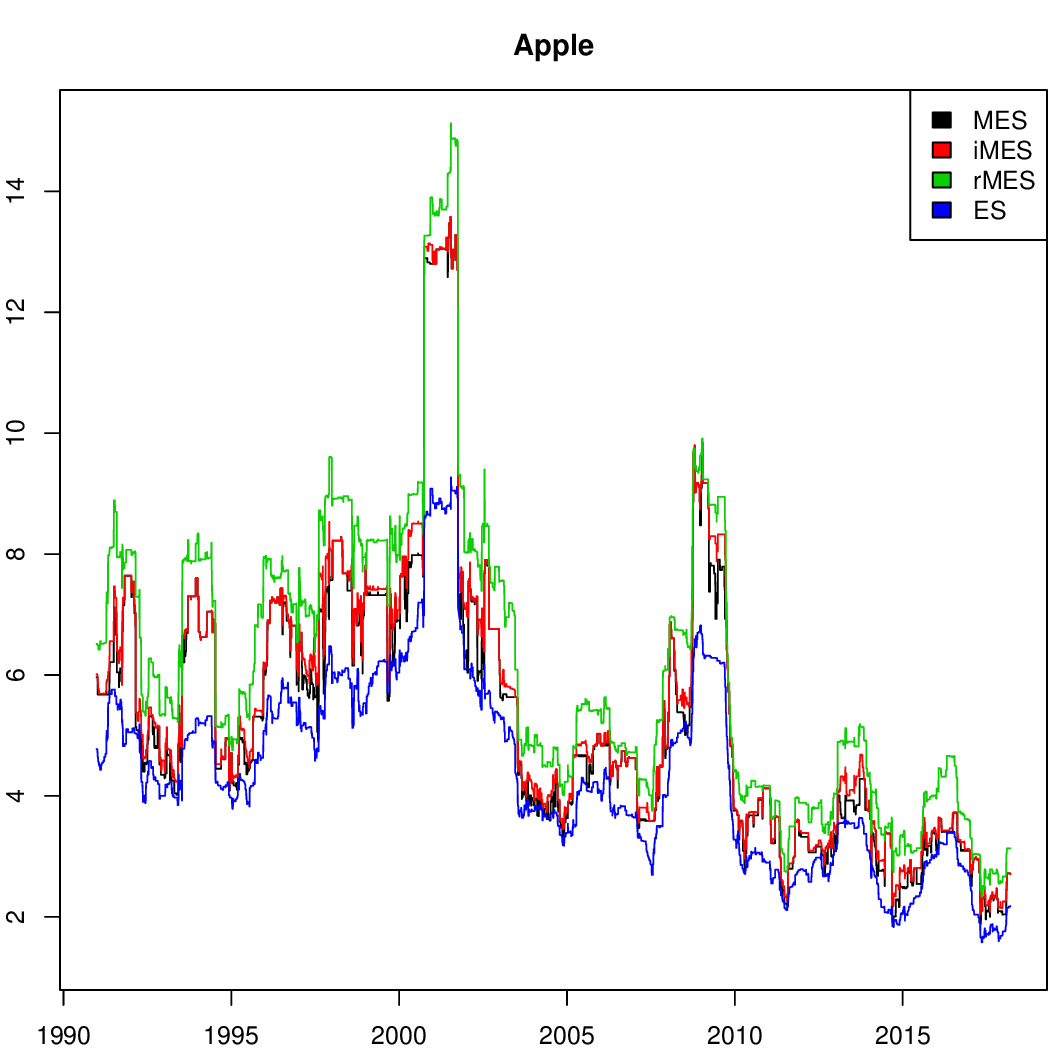}
\includegraphics[width=0.48\textwidth]{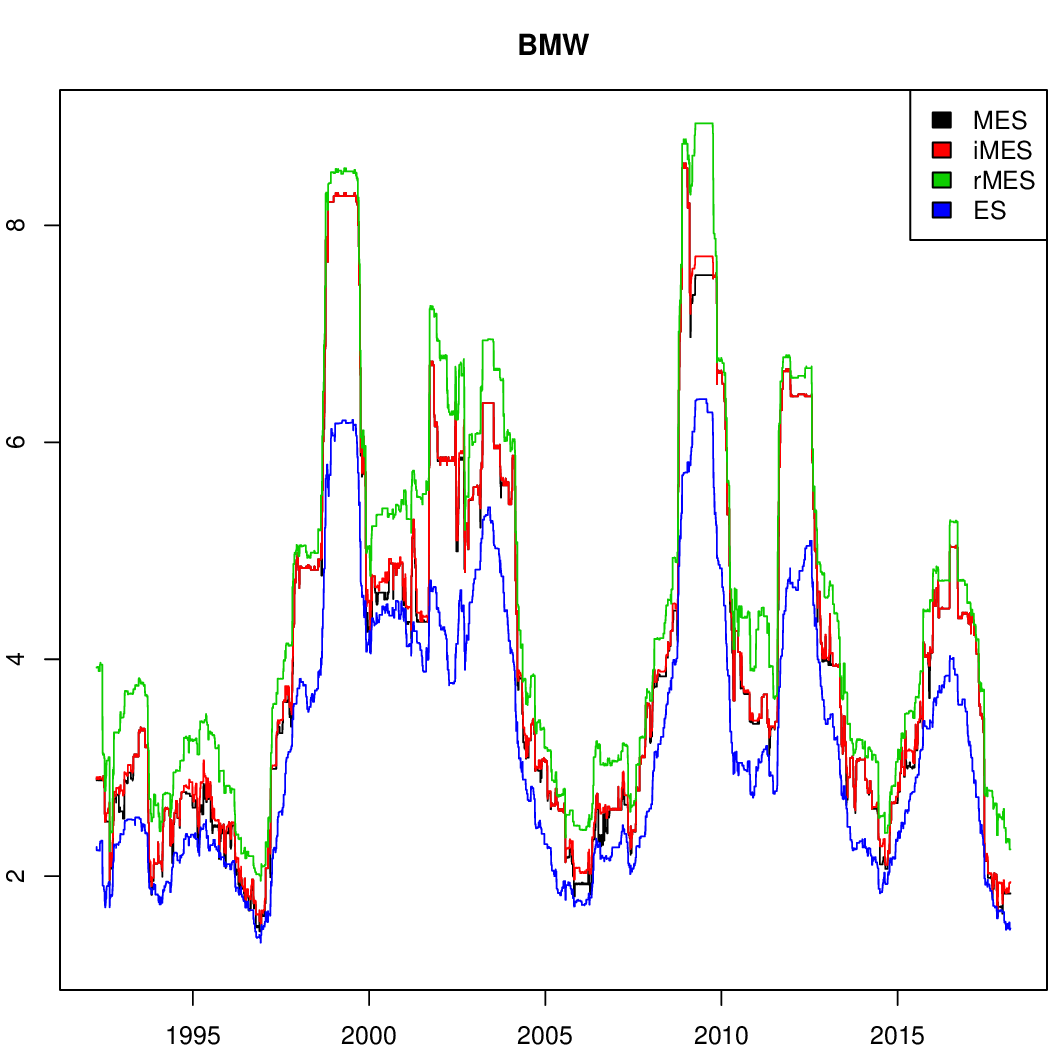}
\end{center}
\caption{$\Q$-based risk measures estimated for data based on economic scenarios with $w = 250$. \label{fig:2}}
\end{figure}

The risk measures $\MES_p$ and $\GES_p$ generally yield similar values. This may be a hint that, often, the empirical distributions under the different scenarios are stochastically ordered. One can observe that during times of financial stress, the $\Q$-based Expected Shortfalls and $\ES_p$ deviate substantially, whereas they are closer during an economically stable period. For the indices (NASDAQ and DAX) $\GES_p$ and $\VES_p$ are closer than for the stock returns (Apple and BMW). The risk measures $\GES_p$ and $\VES_p$ are close if the empirical distribution under one scenario strongly dominates the others in the sense that the quantile function under one scenario is much larger than under the others. 
Therefore, this phenomenon may be explained by the fact that the indices are more closely related to the quantities defining the economic scenarios (VIX and S\&P 500). It also explains why the ratio between $\VES_p$ and $\GES_p$ is generally larger than during financial stress than in economically stable periods. The ratio between $\MES_p$ and $\ES_p$ qualitatively distinguishes the early 2000s recession from the 2008 financial crisis being larger during the latter event, except for the Apple stock. Apple seems to have been more influenced by the dot-com crash in 2000 than the other stocks and indices.

\subsection{The Basel stress-adjustment for Expected Shortfall}\label{sec:52}

In this section, we calculate the stress-adjustment for Expected Shortfall in the Basel market risk evaluation as outlined in Section \ref{sec:11}.
Suppose that there are $n$ securities in a portfolio, and let $P_t^i$, $i=1,\dots,n$, $t\in \N$ denote the time-$t$ price of security $i$.
Let $X_t^i= -(P_t^i/P_{t-1}^i-1) $ be its daily negative return.
Construct a portfolio with price process
$V_t = \sum_{i=1}^n \alpha_i P_t^i  $ where for $i=1,\dots,n$, $\alpha_i$ is the unit of shares invested in security $i$, which is assumed fixed throughout the investment period.
At time $t-1$, we need to to calculate the empirical ES of the next day loss of this portfolio.
Note that  the daily loss is $$V_{t-1} -V_{t} =\sum_{i=1}^n \alpha_i (P^i_{t-1}-P^i_{t} ) =\sum_{i=1}^n X_t^i \alpha_i P_{t-1}^i .$$
At time $t-1$, the values $\alpha_i$ and  $P_{t-1}^i$ are known,  and the random risk factors are $(X_t^1,\dots,X_t^n)$.
To calculate the ES over the data of the past 12-month of data, we need to evaluate the quantity, given the number $\alpha_i P_{t-1}^i$,
$$
\ES^P_p(V_{t-1} - V_{t})=\ES^P_p\Big(\sum_{i=1}^n X_t^i \alpha_i P_{t-1}^i\Big),
$$
where $p=0.975$ as specified in \cite{BASEL16}.
For this purpose,  the scenario $P$ is modelled such that  the distribution of  $(X_t^1,\dots,X_t^{n})$ is according to its empirical version over the past 250 observations, i.e.~over the period $[t-250,t-1]$.   

ES should be calibrated to the most severe 12-month period of stress over a long observation horizon, which has to span back to 2007, as  specified in \cite{BASEL16}. To mimic this adjustment for the period before the introduction of Basel III, it seems fair for everyday evaluation to look back 10 years, and find the maximum ES over a 12-month period. For this purpose, we evaluate, while treating $\alpha_i P_{t-1}^i$ as a constant, the Max-ES given by
$$
\MES^\Q_p(V_{t-1}-V_t) = \MES^\Q_p\Big(\sum_{i=1}^n X_t^i \alpha_i P_{t-1}^i\Big)  = \max_{j=1,\dots,N} \ES^{Q_j}_p\Big(\sum_{i=1}^n X_t^i \alpha_i P_{t-1}^i\Big),
$$
where  $N=2251$, $\Q=\{Q_j\}_{j=1,\dots,N}$, and under $Q_j$, $(X_t^1,\dots,X_t^n)$ is distributed according to its empirical distribution over the time period $[t-j-249, t-j]$. Using the same scenarios, we also compute the integral Max-ES, that is $\GES^\Q_p(V_{t-1}-V_t)$. It is not sensible to compute the replicated Max-ES, $\VES^\Q_p$ in this case, since the number of scenarios is too large. We choose $\alpha_1,\dots,\alpha_n$ such that each $\alpha_i P_i^t $ starts from \$1.
We construct a US stocks portfolio (Apple and Walmart) and a German stocks portfolio (BMW and Siemens).

In Figure \ref{fig:4}, we report $\ES^P_p$, $\MES^\Q_p$ and $\GES^\Q_p$ in the left panel, and 
the percentage of $\ES^P_p$, $\MES^\Q_p$ and $\GES^\Q_p$, that is, $\ES^P_p\left(V_{t-1}-V_t\right)/V_{t-1}$ and so forth, in the right panel. The first row concerns the US stocks portfolio and the second row the German stocks portfolio. 

The percentage $\MES^\Q_p$ and $\GES^\Q_p$ are relatively stable (between 6\% and 9\%, and 7\% and 10\%, respectively), and the percentage $\ES^P_p$ is changing drastically (between 2\% and 9\%), very much depending on the performance of the individual stocks over the past year. This suggests that $\MES^\Q_p$ and $\GES^\Q_p$ have the advantage of being more robust, since they are computed as worst cases over many past scenarios.
The US portfolio has a quite high percentage $\MES^Q_p$ and $\GES^\Q_p$ until 1998 and this is due to the effect of the Black Monday (Oct 19, 1987) that wears out after 10 years. The quotient of $\MES^\Q_p$ and $\GES^\Q_p$ is varies little over time, which is in contrast to the analysis in Section \ref{sec:51}. Currently, we are lacking a clear theoretical explanation of this phenomenon.

If the regulatory capital for the market risk is calculated via $\ES^P_p$, then both portfolios exhibit serious under-capitalization right before the 2007 financial crisis, and their $\ES^P_p$ values increased drastically when the financial crisis took place. On the other hand, if $\MES^\Q_p$ or $\GES^\Q_p$ are used for regulatory capital calculation, then the requirement of capital for both portfolios only increased moderately during the financial crisis. From the data analysis, do not see a clear advantage of $\MES^\Q_p$ over $\GES^\Q_p$, or vice versa. However, $\GES^\Q_p$ seems preferable from a theoretical perspective since it is comonotonic-additive. 
\begin{figure}
\includegraphics[width=0.48\textwidth]{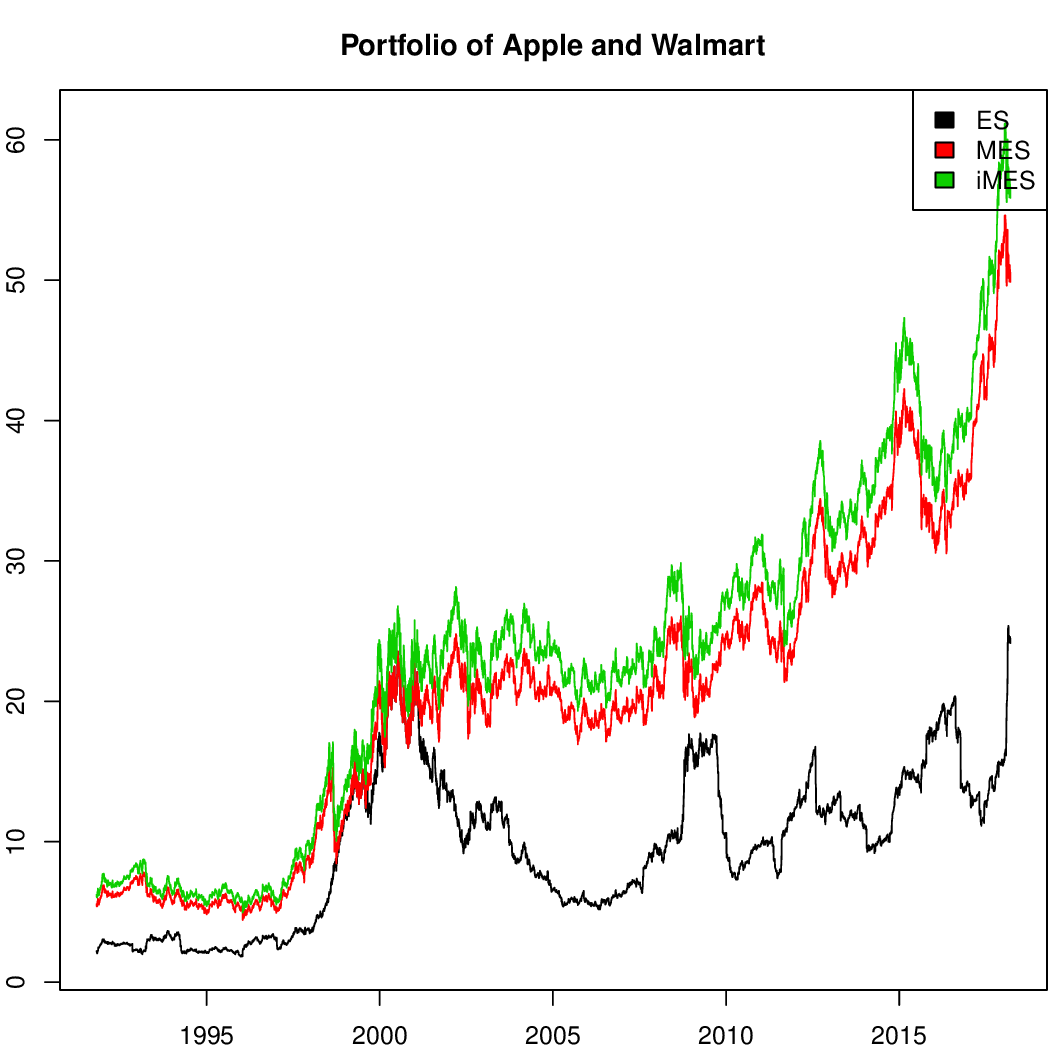}
\includegraphics[width=0.48\textwidth]{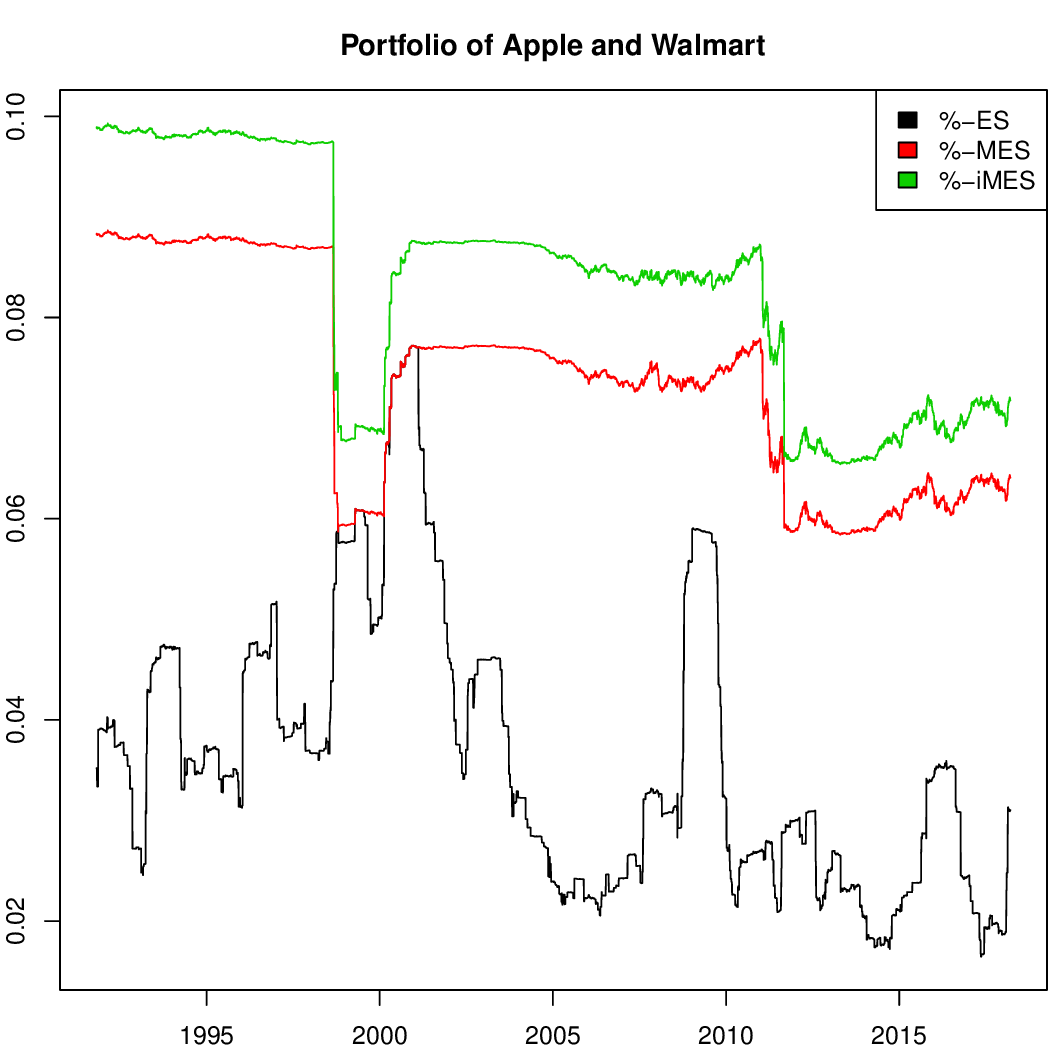}\\
\includegraphics[width=0.48\textwidth]{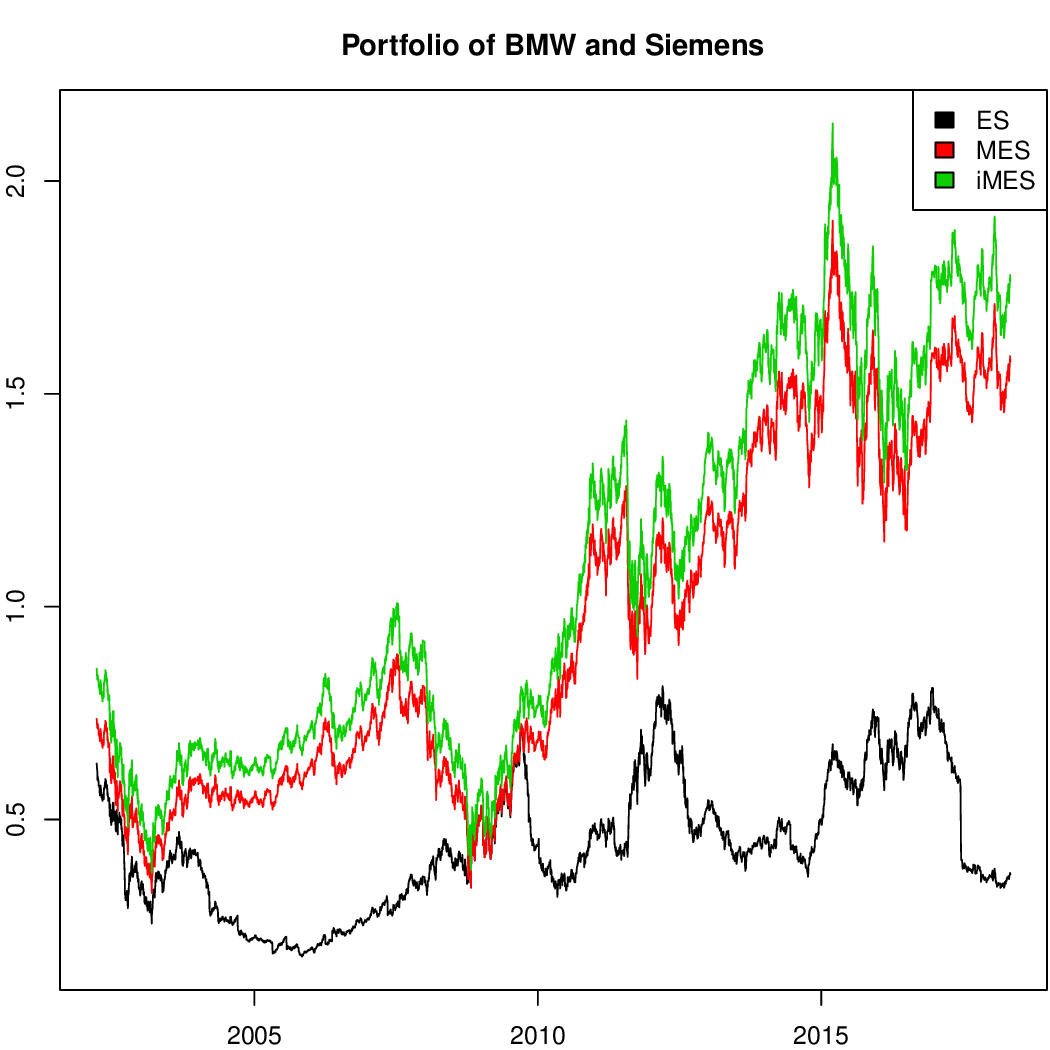}
\includegraphics[width=0.48\textwidth]{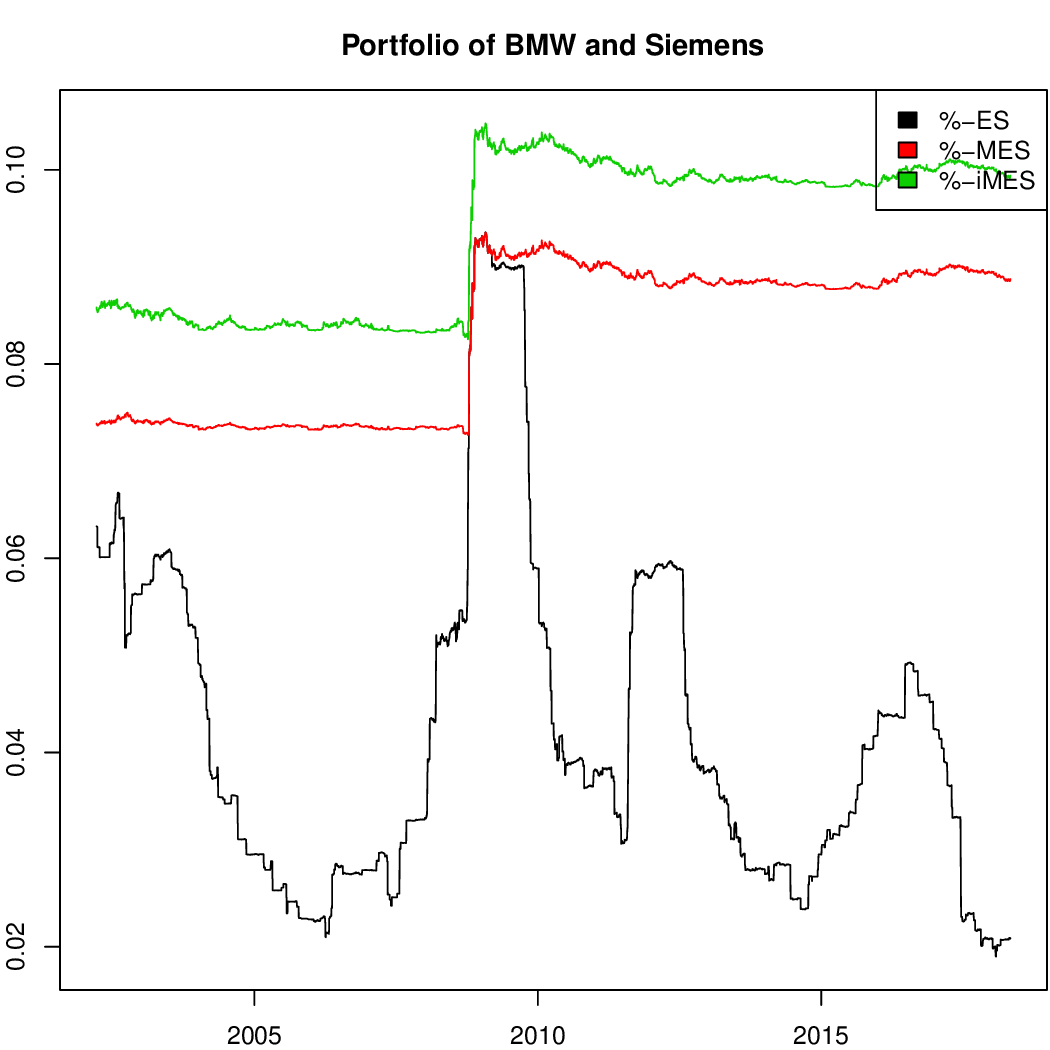}
\caption{The MES and ES of the US and German portfolios. Left panel: MES and ES of the portfolio. Right panel: the percentage of MES and ES in the value of the portfolio.}
\label{fig:4}
\end{figure}

\section{Concluding remarks}

In this paper, we proposed a framework for scenario-based risk evaluation, where different scenarios (probability measures or models) are incorporated into the procedure of risk calculation.
Our framework allows for flexible interpretation of the scenarios and it is in particular motivated by the Basel calculation procedures for the Expected Shortfall, the Chicago Mercantile Exchange, and the credit ratings, as explained in Section \ref{sec:11}.
Several theoretical contributions are made. We introduced the new classes of risk measures including Max-ES, Max-VaR and their variants, and studied their theoretical properties.
Axiomatic characterization of  scenario-based comonotonic additive and coherent classes of risk measures are obtained, and they are well connected to the Basel formulas for market risk.
Finally, we presented  data analyses to illustrate how scenario-based risk measures can be estimated, computed, and interpreted.

Given the pivotal importance of model uncertainty and scenario analysis in modern risk management, scenario-based risk measures can be useful in many disciplines of risk assessment, not limited to financial risk management.

We remark that for various interpretations of the scenarios, the estimation procedures of a scenario-based risk measure may exhibit different properties, as illustrated in Section \ref{sec:6}. This calls for future research in statistical theory for scenario-based risk functionals.
A challenging open question  is the characterization of scenario-based coherent risk measures for general scenarios without assuming mutual singularity.

\begin{acknowledgements}
The authors thank the Editor, an associate editor, two anonymous referees,  Rama Cont, Paul Embrechts, Damir Filipovic, Steven Kou, Fabio Maccheroni, Marco Maggis, Ilya Molchanov and Andreas Tsanakas for helpful discussions on an earlier version of the paper. Wang acknowledges financial support from the Natural Sciences and Engineering Research Council of Canada (RGPIN-2018-03823/RGPAS-2018-522590).
\end{acknowledgements}

\appendix\normalsize

\section{Appendix}
\subsection{Examples and counter-examples}\label{sec:appendix}

  \begin{example}[$\MES_p^\Q$ is not comonotonic-additive]\label{ex:mes-nca}
Let $p\in (0,1)$, and take $Q_1,Q_2\in \mathcal P$, $A_1, A_2\in \mathcal F$ such that
$A_1\subset A_2$, $Q_1[A_1]>Q_2[A_1]$ and also $Q_1[A_2]<Q_2[A_2]<1-p.$ The existence of such $Q_1,Q_2,A_1,A_2$ can be justified by taking $(\Omega,\mathcal F, Q_1)$ and $(\Omega,\mathcal F, Q_2)$ as atomless probability spaces.
Define the set $\Q=\{Q_1,Q_2\}$, $X=\id_{A_1}$ and $Y=\id_{A_2}$. It is clear that $X$ and $Y$ are comonotonic.
Recall that for a  Bernoulli random variable $Z$ under $Q$ with parameter $q$, we have 
$\ES_p^Q(Z)=q/(1-p)$. 
We have
\begin{align*}
\ES_p^{Q_1}(X+Y)&=\ES_p^{Q_1}(X)+\ES_p^{Q_1}(Y) =\frac{1}{1-p}(Q_1[A_1]+Q_1[A_2])
\\&< \frac{1}{1-p}(Q_1[A_1]+Q_2[A_2])=\max_{Q\in \Q} \ES_p^{Q}(X)  + \max_{Q\in \Q} \ES_p^{Q}(Y),
\end{align*}
and similarly,
\begin{align*}
\ES_p^{Q_2}(X+Y) < \max_{Q\in \Q} \ES_p^{Q}(X)  + \max_{Q\in \Q} \ES_p^{Q}(Y)=\MES_p^\Q(X)+\MES_p^\Q(Y).
\end{align*}
Then we have
$$\MES_p^{\Q}(X+Y)=\max\{\ES_p^{Q_1}(X+Y), \ES_p^{Q_2}(X+Y)\}<  \MES_p^\Q(X)+\MES_p^\Q(Y).$$
Thus, $\MES_p^\Q$ is not comonotonic-additive.
  \end{example}

  \begin{example}[$\MES^\Q_p(X)<\ES^\p_p(X)$ for $\Q$ in Example \ref{ex:ex21}] \label{ex:mes-es-ct}
  Consider the set $\Omega=\{\omega_1,\dots,\omega_8\}$ with eight pairwise distinct elements, and let $\p$ be the uniform probability measure on $\Omega$.
  Write $\Omega_1=\{\omega_1,\dots,\omega_4\}$ and $\Theta=\id_{\Omega_1}$.
Let  $Q_1[\cdot] =\p[\cdot|\Theta=1]$,  $Q_2[\cdot] =\p[\cdot|\Theta=0]$ and
 $X=\id_{\Omega_1} + 2\times \id_{\{\omega_8\}}.$
 It is easy to see that $\ES_p^\p(X)=1.25$ and $\ES_p^{Q_1} (X)=\ES_p^{Q_2}(X)=1$.
 Thus, $\MES^\Q_p(X)<\ES^\p_p(X)$.

  \end{example}

    \subsection{Extension of Proposition \ref{th:th32}}
    \label{app:pr5}
    We show that parts (i), (ii) and (iv) of Proposition \ref{th:th32} also hold if $\Q$ is infinite. 
    
    Concerning (i), $\ES_p^Q$ is coherent for $Q\in \Q$.  Since $\MES_p^Q$ can be written as a supremum of coherent risk measures, and taking a supremum preserves all properties of coherent risk measures, $\MES_p^\Q$ is also coherent. An example showing that  $\MES_p^\Q$ is not comonotonic-additive is given in Example \ref{ex:mes-nca}.
Concerning (ii), $\VaR_p^Q$ is monetary for $Q\in \Q$, and hence $\MVaR_p^\Q$,  as a supremum of monetary risk measures, is monetary.
  It remains to show that  $\MVaR_p^\Q$ is a comonotonic-additive risk measure. Using Denneberg's lemma \cite{D94}, for comonotonic random variables $X$ and $Y$, there exist  increasing continuous functions $f$ and $g$ such that $X=f(X+Y)$ and $Y=g(X+Y)$.
  Therefore, for any $Q\in \Q$, we have
 \begin{align*}
 \MVaR_p^\Q(X)=\sup_{Q\in \Q}\VaR_p^Q(f(X+Y))&= \sup_{Q\in \Q}f(\VaR_p^Q(X+Y))\\&=f\Big(  \sup_{Q\in \Q} \VaR_p^Q(X+Y)\Big),
 \end{align*}
 and similarly,
 $$\MVaR_p^\Q(Y)=g\Big(  \sup_{Q\in \Q} \VaR_p^Q(X+Y)\Big).$$
 {Note that $f(z)+g(z)=z$ for $z$ in the range of $X+Y$,
 and by continuity, $f$ and $g$ also satisfy $f(z)+g(z)=z$ for $z=\sup_{Q\in \Q} \VaR_p^Q(X+Y)$.}
Hence, we have
\begin{align*}
\MVaR_p^{\Q}(X+Y)&=  \sup_{Q\in \Q} \VaR_p^Q(X+Y) \\&= f\Big(  \sup_{Q\in \Q} \VaR_p^Q(X+Y)\Big) + g\Big(  \sup_{Q\in \Q} \VaR_p^Q(X+Y)\Big)
\\&= \MVaR_p^\Q(X)+\MVaR_p^\Q(Y).
\end{align*}
  The statement that $\MVaR_p^\Q$ is not necessarily coherent comes from the well-known fact that $\VaR_p^Q$ is not coherent for any $Q\in \mathcal P$ such that $(\Omega,\mathcal F,Q)$ is atomless. Concerning (iv), it suffices to note that ${\GES}_p^\Q$ is a mixture of comonotonic-additive risk measures, and hence it is co\-mo\-no\-to\-nic-additive.

\end{document}